\documentclass[letterpaper,11pt]{article}
\usepackage{amsmath, amsthm, amssymb}
\usepackage[usenames, dvipsnames]{color}
\usepackage[normalem]{ulem} 
\usepackage{fullpage}
\usepackage[numbers, sort]{natbib}
\usepackage[noend]{algorithmic}
\usepackage{tikz, pgfplots}
\usepackage{enumerate}
\usepackage{hyperref}
\usepackage{xspace,color}
\usepackage{graphicx}
\usepackage{caption}
\usepackage{subcaption}
\usepackage{enumitem,linegoal}
\usepackage[linesnumbered,ruled,vlined]{algorithm2e}
\usepackage{comment}	
\usepackage{ctable}					
\newtheorem{claim}{Claim}[section]

\newtheorem{assumption}{Assumption}
\usepackage{mathtools}
\usepackage[flushleft]{threeparttable}
\usepackage{verbatim}
\usepackage{setspace}
\usepackage{bbm}
\usepackage{thm-restate}
\usepackage{footnote}
\makesavenoteenv{tabular}
\makesavenoteenv{table}

\usepackage[margin=1in]{geometry}

\definecolor{crimsonglory}{rgb}{0,0,0}

 \newtheorem{theorem}{Theorem}[section]
 \newtheorem{conjecture}{Conjecture}[section]
 \newtheorem{lemma}[theorem]{Lemma}
 
 \newtheorem{corollary}[theorem]{Corollary}
 
 \newtheorem{definition}[theorem]{Definition}
 \newtheorem{problem}[theorem]{Problem}

\makeatletter
\def\GrabProofArgument[#1]{ #1: \egroup\ignorespaces}
\def\proof{\noindent\textbf\bgroup Proof%
	\@ifnextchar[{\GrabProofArgument}{. \egroup\ignorespaces}}

\makeatother

\usetikzlibrary{arrows,shapes,snakes,automata,backgrounds,petri,calc}
\usepackage[latin1]{inputenc}

\newcommand{\efx}{\textsf{EFX}}
\newcommand{\shortcite}{\cite}
\newcommand{\ech}{\mathsf{H}}
\newcommand{\echpi}{\ech_p}
\newcommand{\graph}{G}
\newcommand{\graphh}{G'}
\newcommand{\ppart}{V}
\newcommand{\vertex}{v}
\newcommand{\ppartt}{U}
\newcommand{\vertexx}{u}
\newcommand{\pparttt}{W}

\newcommand{\mywlog}{without loss of generality}
\newcommand{\oneway}{one-way}
\newcommand{\dep}{corresponding}
\newcommand{\Wlog}{Without loss of generality}
\newcommand{\reachable}{reachable}
\newcommand{\rightreachable}{rightward-reachable}
\newcommand{\rreachable}{rightward-reachable}
\newcommand{\reachableset}{C}
\newcommand{\etal}{\textit{et al. }}
\newcommand{\rhat}{\hat{\reachableset}}
\newcommand{\rainbow}{\mathsf{R}}
\newcommand{\rainbowperm}{\mathsf{R}_p}
\newcommand{\bothomega}[2] {\Phi_{#1, #2}}
\newcommand{\partomega}[1] {\Phi_{#1,*}}
\newcommand{\domega}[1] {\Phi_{*, #1}}

\newcommand{\bothpi}[2] {\Pi_{#1, #2}}
\newcommand{\partpi}[1] {\Pi_{#1,*}}
\newcommand{\dpi}[1] {\Pi_{*, #1}}

\newcommand{\bigvertex}[1]{v_{#1}}

\newcounter{proccnt}

\newcommand{\konote}[1]{}

\usepackage[normalem]{ulem} 

\title{Rainbow Cycle Number and EFX Allocations: \\(Almost) Closing the Gap}

\author{
    Shayan Chashm Jahan
    \and
    Masoud Seddighin
	\and Seyed-Mohammad Seyed-Javadi
	\and Mohammad Sharifi
}

\SetCommentSty{mycommfont}

\SetKwInput{KwInput}{Input}                
\SetKwInput{KwOutput}{Output}  
\begin{document}
	\newcommand{\ignore}[1]{}
\renewcommand{\theenumi}{(\roman{enumi}).}
\renewcommand{\labelenumi}{\theenumi}
\sloppy
\date{} 
\newenvironment{subproof}[1][\proofname]{
	  \renewcommand{\Box}{ \blacksquare}%
	\begin{proof}[#1]%
	}{%
	\end{proof}%
}
\newenvironment{subproof2}{
	\renewcommand{\Box}{ \blacksquare}%
	\begin{proof}%
	}{%
	\end{proof}%
}
\maketitle

\thispagestyle{empty}
\allowdisplaybreaks
\begin{abstract}
Recently, some studies on the fair allocation of indivisible goods notice a connection between a purely combinatorial problem called the Rainbow Cycle problem and a fairness notion known as $\efx$:  assuming that the rainbow cycle number for parameter $d$ (i.e. $\rainbow(d)$) is $O(d^\beta \log^\gamma d)$, we can find a $(1-\epsilon)$-$\efx$ allocation with $O_{\epsilon}(n^{\frac{\beta}{\beta+1}}\log^{\frac{\gamma}{\beta +1}} n)$ number of discarded goods \cite{chaudhury2021improving}.
The best upper bound on $\rainbow(d)$ is improved in a series of works to $O(d^4)$ \cite{chaudhury2021improving}, $O(d^{2+o(1)})$ \cite{berendsohn2022fixed}, and finally to $O(d^2)$ \cite{Akrami2022}.\footnote{We refer to the note at the end of the introduction for a short discussion on the result of \cite{Akrami2022}.} Also, via a simple observation, we have $\rainbow(d) \in \Omega(d)$ \cite{chaudhury2021improving}. 

In this paper, we introduce another problem in extremal combinatorics. For a parameter $\ell$, we define the rainbow path degree and denote it by $\ech(\ell)$. We show that any lower bound on $\ech(\ell)$ yields an upper bound on $\rainbow(d)$. Next, we prove that $\ech(\ell) \in \Omega(\ell^2/\log \ell)$ which yields an almost tight upper bound of   $\rainbow(d) \in \Omega(d \log d)$. This in turn proves the existence of $(1-\epsilon)$-$\efx$ allocation with $O_{\epsilon}(\sqrt{n \log n})$ number of discarded goods. 
In addition, for the special case of the Rainbow Cycle problem that the edges in each part form a permutation, we improve the upper bound to $\rainbow(d) \leq 2d-4$. We leverage $\ech(\ell)$ to achieve this bound. 

Our conjecture is that the exact value of $\ech(\ell) $ is $ \lfloor 
\frac{\ell^2}{2} \rfloor -1$. We provide some experiments that support this conjecture. Assuming this conjecture is correct, we have $\rainbow(d) \in \Theta(d)$.

\end{abstract}
\section{Introduction}
Fair allocation of indivisible goods has been an important problem in economics and social choice theory~\cite{Aziz2015,brams1996fair,Steinhaus:first,Dubins:first,brams,Procaccia:first,ghodsi2018fair,Saberi:first,Etkin2007,Halpern2020,Moulin2019,Procaccia2020,Pratt1990,Budish2012,Budish:first,barman2017finding,vossen2002fair} with many applications in the real world.\footnote{See \url{spliddit.org} and \url{www.fairoutcomes.com} for example.} In a fair allocation problem, we have a set of $n$ agents and a set of $m$ indivisible goods, and each agent has a valuation function that represents her utility for receiving each subset of goods. The goal is to allocate the goods to the agents fairly \cite{Baklanov2020,Brams2017,amanatidis2015approximation,Barman2020,Garg2019,garg2020improved,Bouveret:first}. 

A critical challenge in a fair allocation problem is to specify a reasonable notion of fairness that is simultaneously robust and practical. For the classic version of the problem that the resource is a single divisible good, a notion such as envy-freeness\footnote{An allocation is envy-free if each agent prefers her share over the other agents' share.} perfectly satisfies these conditions: it is commonly accepted as a notion that represents fairness, and there are strong guarantees for the existence of envy-free divisions \cite{Su1999}. However, the applicability of this notion decreases significantly when dealing with indivisible goods: even for two agents and one good, envy-freeness can not be guaranteed. In recent years, several relaxations of envy-freeness have been introduced to adopt this notion to the indivisible setting. Among these notions, $\efx$ is widely believed to be the most prominent \cite{caragiannis2016unreasonable,chaudhury2021improving,chaudhury2021little,Amanatidis2020,Berger2021,Chaudhury2020,plaut2018almost}. 

\begin{definition}
	An allocation is $\efx$ ($\alpha$-$\efx$), if for every agents $i$ and $j$, agent $i$ does not envy ($\alpha$-envy)\footnote{For $\alpha<1$, agent $i$  $\alpha$-envies agent $j$, if the value  of $i$ for his bundle is less than $\alpha$ times his value for bundle of agent $j$.} agent $j$ after removal of any good from the bundle of agent $j$.
\end{definition}

\begin{figure}
	\centering
	\includegraphics[scale=0.9]{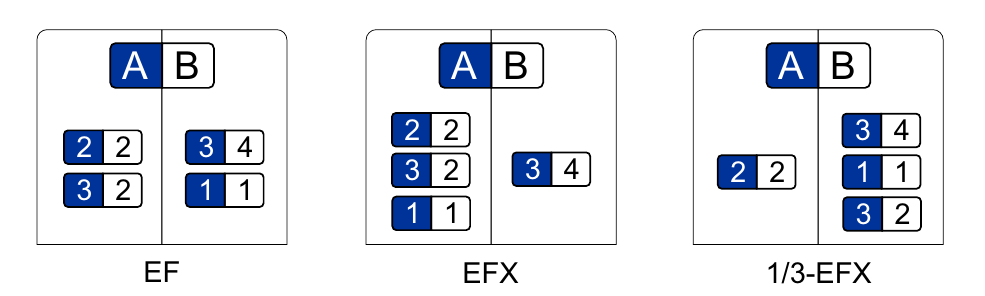}
	\caption{In this figure, three different allocations of four goods to two agents are shown. The valuation of A and B for a good are shown respectively on the left and right sides of the good, and the valuations are additive. The left allocation is envy-free since $v_A(X_A)=2+3=5>v_A(X_B)=3+1=4$ and $v_B(X_B)=4+1=5>v_B(X_A)=2+2=4$. The allocation in the middle is $\efx$ since $v_A(X_A)=6\geq \max_{x\in X_B}v_A(X_B\setminus\{x\})=0$ and $v_B(X_B)=4\geq \max_{x\in X_A}v_B(X_A\setminus\{x\})=4$. Finally, the right  allocation is $1/3$-$\efx$, because $v_A(X_A)=2\geq (1/3) \max_{x\in X_B}v_A(X_B\setminus \{x\})=(3+3)/3=2$ and $v_B(X_B)=7\geq (1/3) \max_{x\in X_A}v_B(X_A\setminus \{x\})=0/3=0$.}
	\label{fig:EF-EFX-alphaEFX-sample}
\end{figure}

See Figure \ref{fig:EF-EFX-alphaEFX-sample} for examples of envy-free, $\efx$, and $\alpha$-$\efx$ allocations. Recent studies suggest that one can obtain strong guarantees on $\efx$ by discarding a subset of goods \cite{chaudhury2021little,caragiannis2019envy}. In a pioneering work, Chaudhury,  Kavitha, Mehlhorn, and Sgouritsa~\shortcite{chaudhury2021little} show that it is possible to find an $\efx$ allocation by discarding at most $n-1$ goods. Further investigations in this direction reveal an intriguing connection between $\efx$ and a purely combinatorial problem called \emph{Rainbow Cycle} problem\footnote{The problem is also known as the Fixed Point Cycle.}~\cite{chaudhury2021improving}. For a multi-partite bidirected graph, a rainbow cycle is a cycle that passes each part at most once. The Rainbow Cycle problem is then defined as follows.

\begin{problem}[Rainbow Cycle]
	\label{rcp}
	For a constant $d$, what is the maximum value $\ell$ such that there exists an $\ell$-partite bidirected graph with no rainbow cycle and the following properties: (i) each part contains at least $d$ vertices, and, (ii) each vertex receives an incoming edge from all other parts other than the one containing it. We call such a value $\ell$ the rainbow cycle number of $d$  and denote it by $\rainbow(d)$.
\end{problem}
We refer to Section \ref{preliminaries} for a more formal definition of this problem.
The connection between the Rainbow Cycle problem and $\efx$ notion was first observed by Chaudhury \etal \shortcite{chaudhury2021improving}: any upper bound on $\rainbow(d)$ yields a corresponding upper bound on the number of discarded goods.

\begin{theorem}[Proved in \cite{chaudhury2021improving}] \label{thm:rainbow2EFX}
	For any constant $\varepsilon \in (0,1/2]$, if there exists constants $\beta, \gamma$ such that $\rainbow(d) \in O(d^\beta \log^\gamma d)$, then we can find a $(1-\varepsilon)$-$\efx$ allocation with $O_{\epsilon}(n^{\frac{\beta}{\beta+1}}\log^{\frac{\gamma}{\beta +1}} n)$ number of discarded goods.
\end{theorem}

The first upper bound on $\rainbow(d)$ was also proposed by Chaudhury \etal \shortcite{chaudhury2021improving}. They proved that $\rainbow(d) \in O(d^4)$ which bounds the number of unallocated goods by  $O_{\epsilon}(n^{\frac{4}{5}})$. Recently, in two parallel studies \cite{berendsohn2022fixed,Akrami2022}, the bound on $\rainbow(d)$ is improved to $O(d^{2+o(1)})$ and $O(d^2)$, yielding an upper bound of $O_{\epsilon}(n^{\frac{2}{3}})$ on the number of unallocated goods. 
Note that a trivial lower bound on $\rainbow(d)$ is $ \Omega(d)$. \footnote{See \cite{chaudhury2021improving} for a matching example.} Therefore, previous results leave a gap of $[\Omega(d), O(d^2)]$ between the best upper bound and the best lower bound. There is a plausible conjecture that $\rainbow(d) \in O(d)$. 

In this paper, we almost close this gap by showing that $\rainbow(d) 
\in O(d\log d)$.
To obtain this bound, we introduce another invariant called rainbow path degree which might be of independent interest. We show that any lower bound on this invariant implies an upper bound on $\rainbow(d)$.  Next, we improve the lower bound on $\rainbow(d)$ by providing an upper bound on the rainbow path degree. 

Before ending this section, we mention that apart from $\efx$ and fair allocation, bounding the rainbow cycles number itself is an interesting extremal problem. Recently, Berendsohn, Boyadzhiyska, and Kozma \shortcite{berendsohn2022fixed} established a connection between two combinatorial problems: Permutation Rainbow Cycle problem which is a special case of Rainbow Cycle problem, and Zero-sum Cycle  problem \cite{Alon2021,Meszaros2021,Alon1993,Alon1993a,Alon1989,Bialostocki1993,Caro1996,schrijver1991simpler}, which is a problem in zero-sum extremal combinatorics. Here we also give an improved upper bound on the permutation rainbow cycle number. We refer to Section \ref{results} for more details on our results and techniques.

\paragraph{A short note on a parallel result.} We note that	parallel and concurrent to this work, Akrami \etal~\shortcite{Akrami2022} also updated their results on arXiv. In  the updated version, their upper bound on $\rainbow(d)$ is improved from $O(d^2)$ to $O(d \log d)$ via a probabilistic  argument. We emphasis that these two studies are parallel and independent.

\section{Preliminaries} \label{preliminaries}

In this paper, our focus is on multi-partite bidirected graphs. For an $\ell$-partite bidirected graph $\graph$, we denote its parts by $\ppart_1,\ppart_2,\ldots,\ppart_{\ell}$.  Also, for a subset $W \subseteq \{\ppart_1,\ppart_2,\ldots,\ppart_{\ell}\}$ we define $\graph[W]$ to be the induced subgraph of $G$ that only includes vertices that belong to the parts in $W$. Thus, $\graph[W]$ has $|W|$ parts.
A path in graph $\graph$ is called rainbow if it passes through each part at most once. The same definition carries over  to cycles. See Figure \ref{fig:rainbow-path-cycle-sample} for an example. 

\begin{figure}
	\centering
	\includegraphics[scale=0.6]{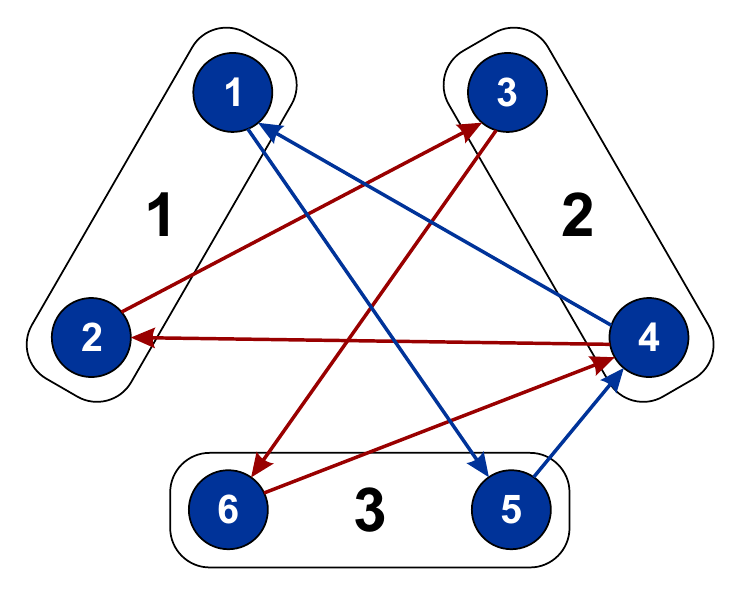}
	\caption{In this figure, path $1\rightarrow5\rightarrow4$ is a rainbow path, and if we add edge $4 \rightarrow 1$ to the end of the path, we have a rainbow cycle. On the other hand, path $2\rightarrow3\rightarrow6  \rightarrow 4$ and cycle $2\rightarrow3\rightarrow6  \rightarrow 4 \rightarrow2$ are not rainbow path and rainbow cycle respectively as they go through part $2$ twice.} 
	\label{fig:rainbow-path-cycle-sample}
\end{figure}

For integers $\ell, d\geq 0$, we define $\bothomega{\ell}{d}$ to be the set of all multi-partite bidirected graphs $\graph$ with the following properties:
\begin{itemize}
	\item $G$ has exactly $\ell$ parts,
	\item each part of $\graph$ has at least $1$ and at most $d$ vertices,
	\item each vertex of $\graph$ has exactly one incoming edge from every other part, 
	\item $\graph$ admits no rainbow cycle.
\end{itemize}
In Figure \ref{fig:sample-omega}, we show an example of a graph in $\bothomega{\ell}{d}$. We also define $\domega{d}$ and $\partomega{\ell}$ as unions of $\bothomega{\ell}{d}$ over all $\ell$ and $d$ respectively, that is, \[\domega{d}=\bigcup_{\ell\geq 0}\bothomega{\ell}{d} \qquad \mbox{and} \qquad \partomega{\ell}=\bigcup_{d \geq 0}\bothomega{\ell}{d}.\] Also, we define $\rainbow(d)$ as the largest $\ell$ such that an $\ell$-partite graph exists in $\domega{d}$, i.e.,
\[
\rainbow(d) = \max_{\ell} \quad \mbox{ s.t. }  \bothomega{\ell}{d} \neq \emptyset.
\]
\begin{figure}
	\centering
	\includegraphics[scale=0.8]{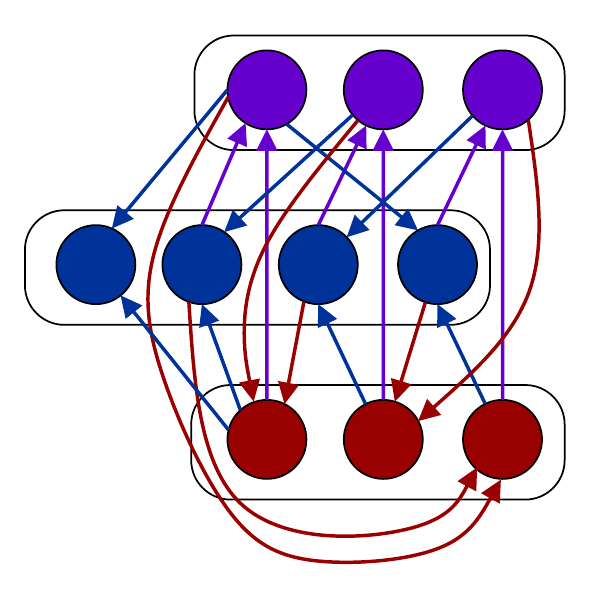}
	\caption{The graph shown in this figure is in $\bothomega{3}{4}$: it contains exactly $3$ parts, each part has at most $4$ vertices, and one can check that two other conditions of $\bothomega{3}{4}$ hold as well. Additionally, by the definition, this graph also belongs to $\partomega{3}$ and $\domega{4}$.} \label{fig:sample-omega}
\end{figure}

Our goal is to give an upper bound on $\rainbow(d)$ for every $d$. To this aim, we introduce another property. 
Let $\graph$ be a multi-partite graph. For every vertex $\vertex \in \graph$, we define $f_\graph(\vertex)$ as the number of vertices in $\graph$ that have a rainbow path to $\vertex$ except $\vertex$ itself. Given $f_\graph(\vertex)$, for every constant $\ell$, we define the rainbow path degree of $\ell$, denoted by $\ech(\ell)$ as follows: 
\[
\ech(\ell) = \min_{\graph \in \partomega{\ell}} \min_{\vertex \in \graph} \quad f_\graph(\vertex).
\]
In other words, $\ech(\ell)$ is the maximum possible value that we are guaranteed that for an $\ell$-partite graph $\graph \in \partomega{\ell}$, for every vertex $\vertex \in \graph$ there are at least $\ech(\ell)$ different vertices with a rainbow path to $\vertex$.  For brevity, we call $\ech(\ell)$ the rainbow path degree of $\ell$.

In order to prove an upper bound on $\rainbow(d)$, we first prove a lower bound on $\ech(\ell)$. Interestingly, though the definition of $\ech(\ell)$ does not depend on $d$, our lower bound on $\ech(\ell)$ results in an almost tight upper bound on $\rainbow(d)$.

In the last part of this section, we mention Stirling's formula for approximating factorials. For every $n>1$, we have:
\begin{equation}
	\sqrt{2\pi n} (\frac{n}{e})^{n} e^{\frac{1}{12n+1}}\leq n!\leq \sqrt{2\pi n} (\frac{n}{e})^{n} e^{\frac{1}{12n}} \label{eq:stirling}.
\end{equation}
In the next section, we briefly  review our results and techniques.


\section{Our Results and Techniques}
\label{results}

The main result of this paper is an almost tight upper bound on the rainbow cycle number by showing that $\rainbow(d) \in \widetilde{O}(d)$. 
Our techniques are structurally different from previous methods. Indeed, a primary application of our techniques provides a simpler proof for $\rainbow(d) \in O(d^2)$. Using a more in-depth analysis, we improve this bound to ${O}(d \log d)$. To show this, we prove a lower bound for the rainbow path degree and show that $\ech(\ell) \in \Omega(\ell^2/\log \ell)$. 
This in turn implies that an $\efx$ allocation exists that discards at most $O_\epsilon(\sqrt{n \log n})$ goods.

For a better understanding of our techniques, let us overview a simple proof for $\rainbow(d) \in O(d^2)$.\footnote{We emphasize that in the interest of simplicity, our discussion in this section is not completely accurate.} 
We prove this bound by showing that $\ech(\ell) \in \Omega(\ell \sqrt{\ell})$. 
Let $\graph \in \partomega{\ell+1}$ be an $\ell+1$ partite graph with parts
$\{\ppart_1,\ppart_2,\ldots,\ppart_{\ell+1}\}$ and let $\vertex$ be a vertex in $\ppart_{\ell+1}$.
By definition, we know that there are at least $\ech(\ell+1)$ vertices that have a rainbow path to $v$.
Denote the set of these vertices by  $\reachableset$. Our goal is to provide a lower bound on $|\reachableset|$. Since the vertices in $\reachableset$ belong to parts $\ppart_1,\ppart_2,\ldots,\ppart_\ell$, there is a part that contributes at most $\ech(\ell+1)/\ell$ vertices to $\reachableset$.
\Wlog, suppose that this part is $\ppart_\ell$. Therefore, 

\[|\ppart_\ell \cap \reachableset| \leq \ech(\ell+1)/\ell.\]
In other words, at most $\ech(\ell+1)/\ell$ of the vertices in $\ppart_\ell$ have a rainbow path to $\vertex$.
Now, consider the vertices that have an outgoing edge to $\vertex$. Since $\graph \in \partomega{\ell+1}$, by definition, each part has a vertex with an outgoing edge to $\vertex$. For each part $\ppart_i$, we assume that $\vertex_i$ is the vertex with an outgoing edge to $\vertex$.
Also, note that for every $1 \leq i \leq \ell-1$, vertex $\vertex_i$ has an incoming edge from part $\ppart_{\ell}$. Since $\vertex_i$ has an outgoing edge to $\vertex$, any vertex in $\ppart_\ell$ that has an outgoing edge to $\vertex_i$ has a rainbow path of length $2$ to $\vertex$ and thus belongs to $\reachableset$. Since $|\ppart_\ell \cap \reachableset| \leq |\reachableset|/\ell$, there exists a vertex $\vertexx \in|\ppart_\ell \cap \reachableset|$ that has outgoing edges to at least
\[(\ell-1)/(\ech(\ell+1)/\ell) \simeq \ell^2/\ech(\ell+1)\]
vertices in $\{\vertex_1,\vertex_2,\ldots,\vertex_{\ell-1}\}$. Denote these vertices by $\rhat$ and suppose \mywlog~that $\vertex_{\ell-1} \in \rhat$. We know that in $G[\ppart \setminus \{\ppart_{\ell-1}, \ppart_{\ell+1}\}]$, the number of vertices that have a rainbow path to $\vertexx$ is at least $\ech(\ell-1)$.
These vertices also have a rainbow path to $\vertex$: consider their rainbow path to $\vertexx$, then go to $\vertex_{\ell-1}$ and then to $\vertex$. Also, these vertices do not belong to $\rhat$; otherwise, since $\vertexx$ has outgoing edges to the vertices in $\rhat$, we have a rainbow cycle. Therefore, 
\begin{equation}\label{eq1}
	\ech(\ell+1) \geq \ell^2/\ech(\ell+1) + \ech(\ell-1).
\end{equation}
Using straightforward calculus one can show that Inequality \eqref{eq1} implies $\ech(\ell+1) \in \Omega(\ell\sqrt \ell)$. 

A consequence of this lower bound is an upper bound on $\rainbow(d)$. To see why, assume for simplicity that $\ech(\ell+1) $ is exactly equal to $ \ell \sqrt{\ell}$. We show $\bothomega{d^2+1}{d}$ is empty. To see why, consider a vertex in $\ppart_{d^2+1}$ with a non-zero outgoing degree. By definition of $\ech(d^2+1)$, the number of vertices with a rainbow path to this vertex is at least $d^2\sqrt{d^2} = d^3$, which is equal to the number of vertices in $\{\ppart_1,\ppart_2,\ldots,\ppart_{d^2}\}$. Thus, any outgoing edge from this vertex yields a rainbow cycle.

In Section \ref{sec4}, via a similar but more in-depth analysis, we show that $\ech(\ell) \in \Omega(\ell^2/\log \ell)$. A consequence of this result is the upper bound of $ O(d \log d)$ on the rainbow cycle number, which leaves a gap of $O(\log d)$ factor between the upper bound and the lower bound for the rainbow cycle number. Also, in Section \ref{section:exp}, we show that $\ech(\ell) \in O(\ell^2)$ that leaves a gap of $O(\log \ell)$ factor between  the upper bound and lower bound for the rainbow path degree.

In Section \ref{section:exp}, we represent our experimental results on finding the exact value of $\ech(\ell)$. Our experiments suggest that for small values of $\ell$, we have  $\ech(\ell) = \lfloor{\frac{\ell^2}{2}}\rfloor-1$. Assuming that this conjecture is correct for every $\ell$, we have $\rainbow(d) \in O(d)$. As a future direction, one can think of improving the lower bound on $\ech(\ell)$ to $\Omega(\ell^2)$.

Also, we consider a special case of the Rainbow Cycle problem called the Permutation Rainbow Cycle problem, where each vertex has exactly one outgoing edge to each part. As we mentioned earlier, this problem has some independent applications in extremal combinatorics. We improve the upper bound on the permutation rainbow cycle number to $2d-3$. Next, we leverage the bounds we obtain on $\ech(\ell)$ for small values of $\ell$ in Section \ref{section:exp} to improve the upper bound to $2d-4$.
Furthermore, In Section \ref{section:exp}, we consider the relation between the rainbow cycle number and the rainbow path degree in the permutation case. We show that our conjecture of $\ech(\ell) = \lfloor{\frac{\ell^2}{2}}\rfloor-1$ implies the upper bound of $2d-3$ on $\rainbowperm(d)$ in the permutation case. 


\section{Upper Bound on the Rainbow Cycle Number}\label{sec4}
We now present our results for the Rainbow Cycle problem. This section is divided into three parts. In the first part, in Lemma \ref{MTH2}, we show that any lower bound on rainbow path degree implies a corresponding upper bound on rainbow cycle number. Next, we prove two lower bounds on $\ech(\ell)$. As a warm up, we start by showing that $\ech(\ell) \in \Omega(\ell\sqrt{\ell})$. 
Next, we present the main result of this section, that is, $\ech(\ell) \in \Omega(\ell^2/\log \ell)$. This, combined with Lemma \ref{MTH2}, yields the upper bound of $\rainbow(d) \in O(d \log d)$. 

Lemma \ref{MTH2} shows a simple connection between $\ech(\ell)$ and $\rainbow(d)$. 
The idea behind the proof of Lemma \ref{MTH2} is simple: the rainbow path degree of a vertex cannot be more than the total number of the vertices.
\begin{lemma}\label{MTH2}
	For every $\beta>0$,$\gamma$ if 
$\ech(\ell)\in \Omega(\ell^{1+\beta}\log^\gamma \ell)$ then  $\rainbow(d)\in O(d^{\frac{1}{\beta}}\log^{-\frac{\gamma}{\beta}} d)$.
\end{lemma}
\begin{proof} 
	Since $\ech(\ell)\in \Omega(\ell^{1+\beta}\log^\gamma \ell)$, there exist $a, \ell_0 \in \mathbf{N}$ such that for every $\ell\geq \ell_0$ inequality 
\[\ech(\ell)\geq a\cdot\ell^{1+\beta}\log^\gamma \ell\]
always holds. Also, by definition, we know that $\bothomega{\rainbow(d)}{d} \neq \emptyset$. Let $G \in \bothomega{\rainbow(d)}{d}$ be a graph. Since the number of vertices in $\graph$ is $d\cdot\rainbow(d)$, the total number of vertices that have a rainbow path to any vertex is upper bounded by $d\cdot\rainbow(d)$.
Furthermore, we have $\rainbow(d)\geq d$, and therefore, for $d\geq \ell_0$, 
\[a\cdot\rainbow(d)^{1+\beta}\log^\gamma \rainbow(d)\leq d\cdot\rainbow(d),\]
which means
	\[\rainbow(d)\leq \frac{1}{a^{\frac{1}{\beta}}}d^{\frac{1}{\beta}}\log^{-\frac{\gamma}{\beta}} \rainbow(d).\] 
	
	Since $\rainbow(d) \in [d,d^2]$, we have $\log d \leq \log \rainbow(d)\leq 2\log d$. Thus,
	
	\[\rainbow(d)\leq \frac{2^{|\frac{\gamma}{\beta}|}}{a^{\frac{1}{\beta}}}d^{\frac{1}{\beta}}\log^{-\frac{\gamma}{\beta}} d,\]
	
	which implies
	\[\rainbow(d)\in O(d^{\frac{1}{\beta}}\log^{-\frac{\gamma}{\beta}} d).\]
\end{proof}

We use Lemma \ref{MTH2} to prove two upper bounds on $\rainbow(d)$. First, in Lemma \ref{ind}, we show that $\ech(\ell) \in \Omega(\ell \sqrt{\ell})$, which  implies $\rainbow(d) \in O(d^2)$.
\begin{lemma}\label{ind}
	For every $\ell \geq 1$, we have $ \ech(\ell +1)\geq \ell \sqrt{\ell}/6.$
\end{lemma}

\begin{proof}
	In order to prove Lemma \ref{ind}, we use induction on $\ell$. For $\ell = 1, 2$ we have: 

\[\frac{\ell \sqrt{\ell}}{6} \leq \frac{2\sqrt{2}}{6} < 1 \leq \ech(\ell+1).\]

Now, suppose that the statement holds for every $\ell' < \ell$. Our goal is to prove the claim for $\ell$.
As a contradiction, suppose 
\begin{equation}
	\ech(\ell +1)<\ell \sqrt{\ell}/6.
\end{equation}
This means that there exists a graph $G \in \partomega{\ell+1}$ and a vertex $\vertex \in  G$, such that if we define $\reachableset$ as the set of the vertices in $G$ with a rainbow path to $\vertex$, we have
\begin{equation}
	|\reachableset| < \ell \sqrt{\ell}/6.\label{eq:contradiction_v_lsqrtl}
\end{equation}
Suppose that $\{\ppart_1,\ppart_2,\ldots,\ppart_{\ell+1}\}$ is the set of parts in $G$ and suppose \mywlog~that $\vertex \in \ppart_{\ell+1}$.

\begin{claim}\label{lmmm}
	Fix a vertex $u$, and define $P$ as the set of all rainbow paths with length at most $2$ from $u$ to $\vertex$. Also, let $\rhat$ be the set of all different vertices that have an incoming edge from $u$ and belong to a path in $P$. We have $|\rhat|\leq 2\sqrt{\ell}/3$.
\end{claim}

\begin{subproof}[of Claim \ref{lmmm}]
		Suppose $|\rhat|> 2\sqrt{\ell}/3$ and assume \mywlog~that $u\in \ppart_{\ell}$.
	Consider an arbitrary path $p\in P$ and \mywlog~suppose that $p$ does not pass $\{\ppart_1,\ppart_2,\ldots,\ppart_{\ell-2}\}$ (recall that the length of $P$ is at most 2).
	Now, define $G'=\graph[\{\ppart_1,\ppart_2,\ldots,\ppart_{\ell-2}, \ppart_{\ell}\}]$. Therefore, $G'$ contains $\ell-1$ parts. By the induction hypothesis, we know at least $(\ell-2)\sqrt{\ell-2}/6$ of the vertices in $G'$ have a rainbow path to $u$. Let $S$ be the set of these vertices. Note that all the vertices in $S$ have also a rainbow path to $\vertex$: consider their rainbow path to $u$, then go to $\vertex$ by path $p$. These two paths do not intersect because path $p$ does not pass $\{\ppart_1,\ppart_2,\ldots,\ppart_{\ell-2}\}$. Furthermore, none of these vertices belong to $\rhat$; otherwise, since $u$ has an outgoing edge to the vertices in $\rhat$ we have a rainbow cycle (see Figure \ref{fig:lemma4-3}). Therefore, the number of vertices that have a rainbow path to $v$ is at least
	\begin{figure}
		\centering
		\includegraphics[scale=0.5]{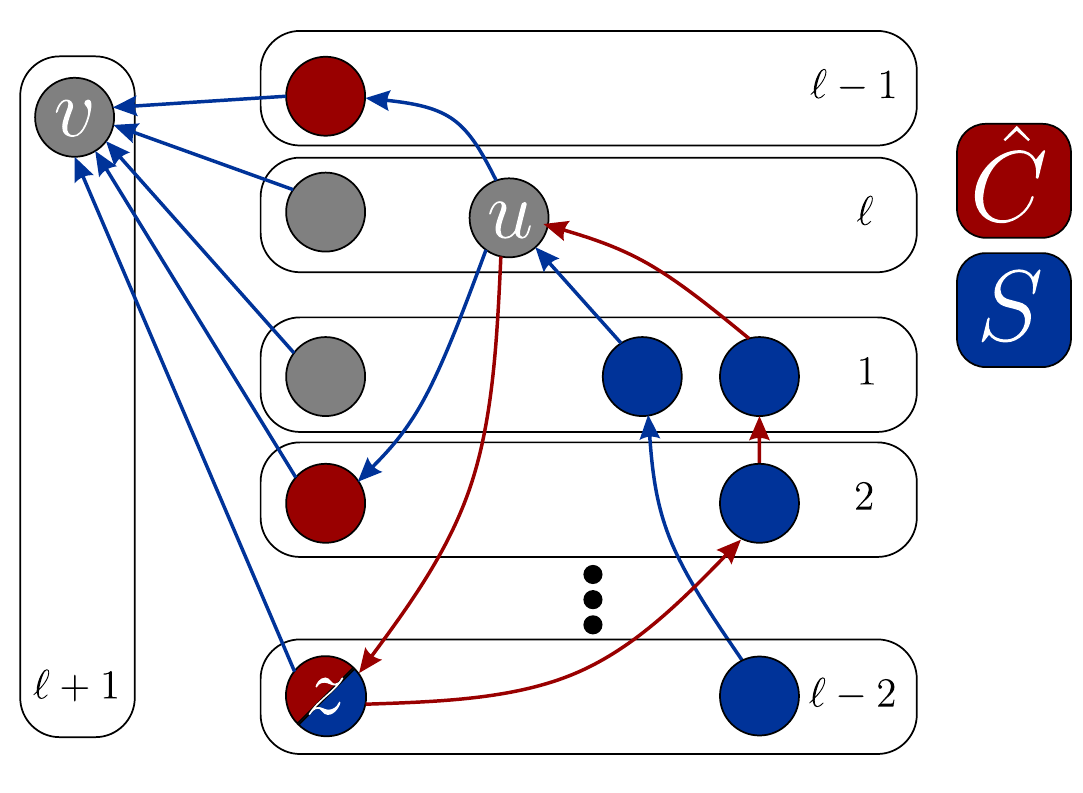}
		\caption{In this figure, sets $\hat{C}$ and $S$ as defined in the proof of Claim \ref{lmmm} are shown. One can observe that if a vertex $z$ belongs to both $\rhat$ and $S$, we have a rainbow cycle.} \label{fig:lemma4-3}
	\end{figure}
	\[\begin{aligned}
		|\reachableset| &\geq |\rhat| + |S| \nonumber\\
		&\geq |\rhat| + \ech(\ell-1) \nonumber\\
		&\geq \frac{2\sqrt{\ell}}{3} + \frac{(\ell-2)\sqrt{\ell-2}}{6} \nonumber\\
		&= \frac{2\sqrt{\ell}}{3} + \frac{\ell\sqrt{\ell-2}}{6} - \frac{\sqrt {\ell-2}}{3} \nonumber\\
		&\geq \frac{\sqrt{\ell}}{3} + \frac{\sqrt\ell\sqrt{\ell^2-2\ell}}{6} \nonumber\\
		&\geq \frac{\sqrt{\ell}}{3} + \frac{\sqrt\ell(\sqrt{\ell^2-2\ell+1}-1)}{6} \nonumber\\
		&= \frac{\sqrt{\ell}}{3} + \frac{\sqrt\ell(\ell-2)}{6}\nonumber\\
		&\geq \frac{\ell \sqrt{\ell}}{6}\nonumber,
	\end{aligned}\]
	which contradicts Inequality (\ref{eq:contradiction_v_lsqrtl}). This completes the proof of Claim \ref{lmmm}.
\end{subproof}

By Inequality (\ref{eq:contradiction_v_lsqrtl}), we know $|\reachableset| < \ell \sqrt{\ell}/6$. The vertices in $\reachableset$ belong to parts $\ppart_{1},\ppart_2,\ldots,\ppart_\ell$. Therefore, at least one of these parts contributes less than $\sqrt{\ell}/6$ vertices to $\reachableset$. Suppose \mywlog~that $\ppart_\ell$ is one of such parts, i.e., $|\ppart_\ell \cap \reachableset| < \sqrt{\ell}/6$.  Since $G \in \partomega{\ell+1}$, each part other than $\ppart_{\ell+1}$ has a vertex with an outgoing edge to $\vertex$. For each part $\ppart_i$ ($i\leq \ell-1$), we denote this vertex by $\vertex_i$.  Also, note that each vertex $\vertex_i$ has an incoming edge from $\ppart_\ell$. Since $\vertex_i$ has an outgoing edge to $\vertex$, any vertex in $\ppart_\ell$ that has an outgoing edge to $\vertex_i$ has a rainbow path of length $2$ to $\vertex$ and thus belongs to $\reachableset$. Hence, at least one of the vertices in $\ppart_\ell \cap \reachableset$ has outgoing edges to at least
\begin{equation}
	\frac{\ell-1}{\sqrt{\ell}/6} = \frac{6(\ell-1)}{\sqrt{\ell}}
\end{equation}
of the vertices in $\{\vertex_1,\vertex_2,\ldots,\vertex_{\ell-1}\}$. On the other hand, by Claim \ref{lmmm}, we know that each vertex in $\ppart_{\ell}$ has at most  $2\sqrt{\ell}/3$ outgoing edges to  $\{\vertex_1,\vertex_2,\ldots,\vertex_{\ell-1}\}$. Thus, we have
\[	\frac{6(\ell-1)}{\sqrt{\ell}} \leq \frac{2\sqrt{\ell}}{3},\]
which means 

\[{18(\ell-1)} \leq {2\ell},\]
that is, $\ell \leq 16/18$. But this contradicts the fact that  $\ell>2$. 	
 \end{proof}

\begin{corollary}[of Lemma \ref{ind}]
	By choosing $\beta = 0.5$ and $\gamma = 0$  in Lemma \ref{MTH2}, we have $\rainbow(d) \in O(d^2)$.
\end{corollary}

Now, we are ready to prove our main result. In Theorem \ref{MTH}, we show that $\ech(\ell+1)  \in \Omega(\ell^2/\log \ell)$. The structure of the proof of Theorem \ref{MTH} is similar to the proof of Lemma \ref{ind}. The difference is that here we generalize Claim \ref{lmmm} to consider paths with length more than 2. 
\begin{theorem}\label{MTH}
	For every $\ell \geq 3$, we have $\ech(\ell +1) \geq \ell^{2}/20\ln\ell$.
\end{theorem}
\begin{proof} 
	We use induction on $\ell$. For $\ell = 3, 4$ we have: 
\[\frac{\ell^{2}}{20\ln\ell} < 1 \leq \ell \leq \ech(\ell+1).\]
Now, suppose that for some $\ell \geq 5$ we know that the statement of Theorem \ref{MTH} holds for every $3 \leq\ell' < \ell$ and our goal is to prove the claim for $\ell$.
As a contradiction, suppose 
\begin{equation}
	\ech(\ell +1)<\frac{\ell^{2}}{20\ln\ell}. \label{eq:contradiction-l2}
\end{equation}
This means that there exists a graph $G \in \partomega{\ell+1}$ and a vertex $\vertex \in  G$, such that exactly $\ech(\ell +1)$ of the vertices in $G$ have a rainbow path to $\vertex$, which is less than $\ell^{2}/20\ln\ell$. Suppose that $\{\ppart_1,\ppart_2,\ldots,\ppart_{\ell+1}\}$ is the set of parts in $G$ and suppose \mywlog~that $\vertex \in \ppart_{\ell+1}$. We start by proving Claim \ref{lmm}. Claim \ref{lmm} plays a similar role as Claim \ref{lmmm}. The main difference is that in Claim \ref{lmm}, we consider paths with length more than  2.

\begin{restatable}{claim}{secondclaim}\label{lmm}
	Fix a vertex $u$ and an integer  $k\leq \ell-3$, and define $P_k$ as the set of all rainbow paths with length at most $k$ from $u$ to $\vertex$. Also, let $\rhat$ be the set of all different vertices that have an incoming edge from $u$ in the paths of $P_k$. Then, $|\rhat|\leq \ell k/4\ln\ell$.
\end{restatable}

	\begin{subproof}[of Claim \ref{lmm}]
			As a contradiction, suppose $|\rhat|\leq\ell k/4\ln\ell$. \Wlog, we can assume that $u\in \ppart_{\ell}$. 
		Consider an arbitrary path $p\in P_k$. Since the length of $p$ is at most $k$, we can suppose \mywlog~that $p$ does not pass through $\{\ppart_1,\ppart_2,\ldots,\ppart_{\ell-k}\}$. 
		Define $S$ as the set of all vertices in  $\graph[\{\ppart_1,\ppart_2,\ldots,\ppart_{\ell-k},\ppart_\ell\}]$ that have a rainbow path to $u$. By definition, we know that $|S|\geq \ech(\ell - k+1)$.
		\begin{figure}
			\centering
			\includegraphics[scale=0.5]{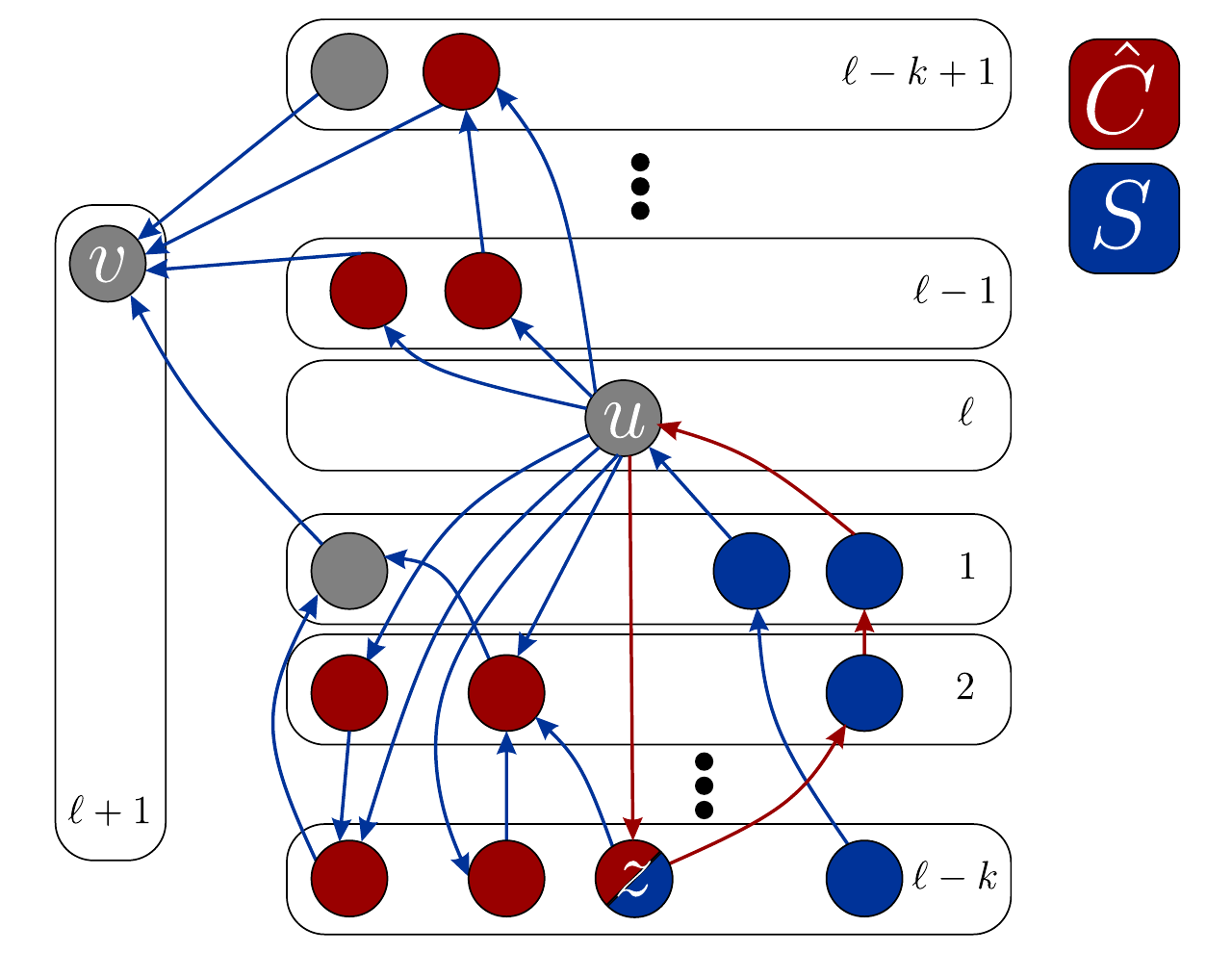}
			\caption{In this figure, two main sets defined in the proof of Claim \ref{lmm} are shown. Furthermore, one can notice that if a vertex $z$ is in both of the sets $\rhat$ and $S$, there is a rainbow cycle like the rainbow cycle depicted in the figure with red edges.} \label{fig:lemma4-6}
		\end{figure}
		Note that $\vertex$, as well as $u$, can be reached by rainbow paths from all the vertices in $S$ since $u$ can be reached by rainbow paths from all of the vertices in $S$ through $\{\ppart_1,\ppart_2,\ldots,\ppart_{\ell-k}\}$, and furthermore, $u$ has a rainbow path to $\vertex$ that does not pass any vertex in parts $\{\ppart_1,\ppart_2,\ldots,\ppart_{\ell-k}\}$. Therefore:
		\[\ech(\ell+1)\geq |\rhat\cup S|.\]
		On the other hand, we have $\rhat\cap S=\emptyset$; otherwise, if there exists a vertex $z \in \rhat\cap S$, we know that $z$ has an incoming edge from $u$ (since $z\in \rhat$) and also $z$ has a rainbow path to $u$ (since $z\in S$), which yields  a rainbow cycle (see Figure \ref{fig:lemma4-6}). Therefore:
		\[\ech(\ell+1)\geq |\rhat\cup S| = |\rhat| + |S| \geq \ech(\ell - k +1) + \frac{\ell k}{4\ln\ell}.\]
		Since $\ell-k \geq 3$,  by  induction hypothesis, we have $ \ech(\ell - k+1)\geq (\ell - k)^2/20\ln(\ell-k)$.
		Thus:
		\[\begin{aligned}
			\ech(\ell+1)&\geq \frac{(\ell - k)^2}{20\ln(\ell-k)} + \frac{\ell k}{4\ln\ell}\\
			&\geq\frac{1}{20\ln\ell}((\ell - k)^2 + 2\ell k)\\
			&\geq\frac{\ell^2}{20\ln\ell},
		\end{aligned}\]
		which contradicts Inequality (\ref{eq:contradiction-l2}). This completes the proof.
	\end{subproof}
	
	\begin{claim}
		\label{clm4}
		Fix an integer $t\leq \ln \ell$. For every  $k\leq t$ and subset $W$ of $\{\ppart_1,\ppart_2,\ldots,\ppart_\ell\}$ with $\ell - t +k$ parts, 
		at least $t^{k-1}/ (k-1)!$ vertices of each part in $W$ have rainbow paths with lengths at most $k$ to $\vertex$ in $G[W\cup \{\ppart_{\ell+1}\}]$.
	\end{claim}

	\begin{subproof}[of Claim \ref{clm4}]
	We prove this claim by induction on $k$. For the base case $k=1$, it is sufficient to show that each part has at least $t^{k-1}/ (k-1)!=1$ vertex which has an outgoing edge to $\vertex$. This is trivial since $\graph \in \partomega{\ell+1}$.
	
	Now, suppose that the statement holds for every $k' < k$. We prove the claim for $k$.  As a contradiction suppose that for set $ W = \{\ppart_1,\ppart_2,\ldots,\ppart_{\ell-t+k}\}$ and part $\ppart_{\ell-t+k}$ we have that less than $t^{k-1}/ (k-1)!$ of the vertices  of $\ppart_{\ell-t+k}$ have a rainbow path with length at most $k$ to $\vertex$ in $\graph[W\cup \ppart_{\ell+1}]$. Define $Q$ as the set of these vertices. We have:
	\begin{equation}\label{cllllm}
		|Q|<\frac{t^{k-1}}{(k-1)!}. 
	\end{equation} 
	By induction hypothesis, in $\graph[\{\ppart_1,\ppart_2,\ldots,\ppart_{\ell-t+k-1}\} \cup \{\ppart_{\ell+1}\}]$, at least 
	$t^{k-2}/ (k-2)!$ of the vertices of each part in $\{\ppart_1,\ppart_2,\ldots,\ppart_{\ell-t+k-1}\}$ have a rainbow path with length at most $k-1$ to $\vertex$. Hence, in total, in $\graph[\{\ppart_1,\ppart_2,\ldots,\ppart_{\ell-t+k-1}\} \cup \{\ppart_{\ell+1}\}]$, at least 
	$(\ell-t+k-1)(t^{k-2}/ (k-2)!)$ of the vertices have  a rainbow path with length at most $k-1$ to $\vertex$. Define $\reachableset^{k-1}$ as the set of these vertices. 
	
	Note that each vertex in $\reachableset^{k-1}$ has an incoming edge from $\ppart_{\ell-t+k}$.  Define $T$ as the set of vertices in $\ppart_{\ell-t+k}$ that have outgoing edges to vertices in $\reachableset^{k-1}$. Since each vertex in $\reachableset^{k-1}$ has a rainbow path with length at most $k-1$ to $\vertex$,  each vertex in $T$ has a rainbow path of length at most $k$ to $\vertex$. Therefore, we have $T\subseteq Q$. As a consequence, Inequality (\ref{cllllm}) implies that $|T|<t^{k-1}/ (k-1)!$. Moreover, there is a vertex $\vertexx\in T$ that has at least 
	\[\frac{|\reachableset^{k-1}|}{|T|} \geq (\ell-t+k-1)\frac{\frac{t^{k-2}}{(k-2)!}}{\frac{t^{k-1}}{(k-1)!}}\]
	outgoing edges to $\reachableset^{k-1}$. By Claim \ref{lmm}, we can conclude that $\vertexx$ has less than $\ell k/4\ln\ell$ outgoing edges to $\reachableset^{k-1}$ \footnote{Since $k\leq t\leq \ln\ell \leq \ell-3$, the conditions of Claim \ref{lmm} hold.}(see Figure \ref{fig:clm4} for an illustration).  Therefore:
	\begin{figure}
		\centering
		\includegraphics[scale=0.5]{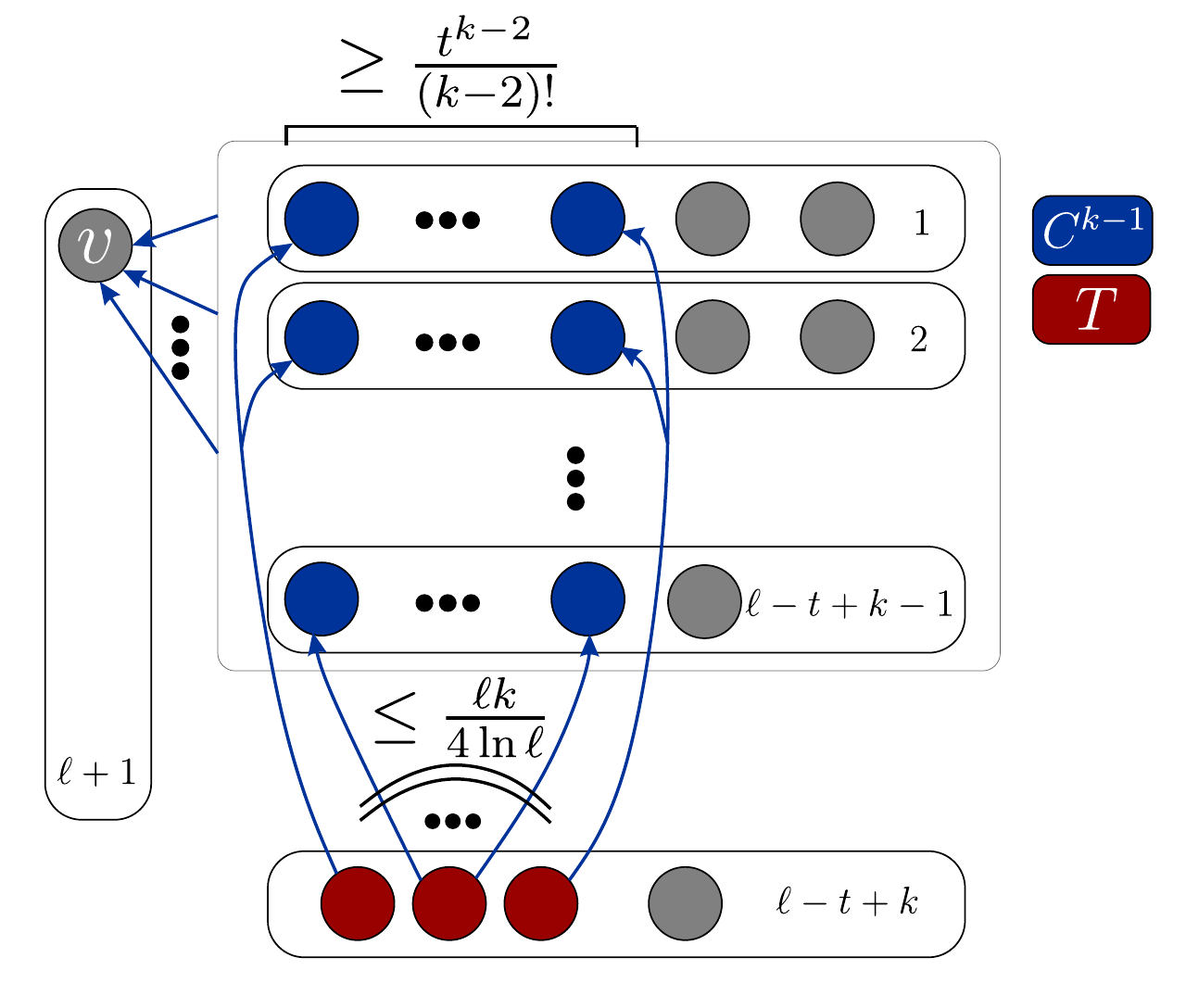}
		\caption{In this figure, you can see an illustration of sets $\reachableset^{k-1}$ and $T$ and the bounds on the size of $\reachableset^{k-1}$ in each part and the number of outgoing edges from $T$ to $\reachableset^{k-1}$ for each vertex. } \label{fig:clm4}
	\end{figure}
	
	\[\begin{aligned}
		\frac{\ell k}{4\ln\ell}&\geq\frac{(\ell-t+k-1)(\frac{t^{k-2}}{(k-2)!})}
		{(\frac{t^{k-1}}{(k-1)!})}\\
		&=\frac{(\ell-t+k-1)(k-1)}{t} \\
		&\geq\frac{(\ell-t)(k/2)}{t}&k\geq 2\\
		&>\frac{(\ell/2)(k/2)}{t}&\ell\geq 5, t\leq\ln\ell\\
		&\geq\frac{\ell k}{4\ln\ell}&t\leq\ln\ell,
	\end{aligned}\]
	which is a contradiction. This completes the proof of Claim \ref{clm4}.
	\end{subproof}

	Note that by the pigeonhole principle, there exists a part $\ppart_i$ that contains at most $\ech(\ell +1)/\ell$ vertices with a rainbow path to $\vertex$. By setting $k=t$ in Claim \ref{clm4}, we have:
	\begin{align*}
		\frac{\ech(\ell +1)}{\ell}&\geq \frac{t^{t-1}}{(t-1)!} \\
		&\geq \frac{t^{t-1}}{(\frac{t-1}{e})^{t-1}\sqrt{2\pi(t-1)}\cdot e^{\frac{1}{12(t-	1)}}} &\mbox{Inequality } \eqref{eq:stirling}\\
		&\geq\frac{e^{t-1}}{e^{\frac{1}{12t}}\sqrt{2\pi t}} =\frac{e^{t+1}}{e^{2+\frac{1}{12t}}\sqrt{2\pi t}}.
	\end{align*}
If we choose $t = \left \lfloor{\ln\ell}\right \rfloor$, we have\footnote{Since $\left \lfloor{\ln\ell}\right \rfloor\leq \ln\ell$, the conditions of Claim \ref{clm4} hold.}:
\begin{align*}
	\frac{e^{t+1}}{e^{2+\frac{1}{12t}}\sqrt{2\pi t}} 
	&\geq\frac{\ell}{e^{2+\frac{1}{12t}}\sqrt{2\pi t}}&t+1\geq\ln\ell  \\
	&\geq\frac{\ell}{e^{2+\frac{1}{12}}\sqrt{2\pi \ln\ell}}&1\leq t \leq\ln\ell\\
	&\geq\frac{\ell}{21\sqrt{\ln\ell}} \\
	&\geq\frac{\ell}{20\ln\ell} &\ell\geq5,
\end{align*}
	which contradicts Inequality \eqref{eq:contradiction-l2}. This Completes the proof of Theorem \ref{MTH}.
\end{proof}
\begin{corollary}[of Theorem \ref{MTH}] \label{cor:thm:MTH}
	By choosing $\beta = 1$ and $\gamma = 1$  in Lemma \ref{MTH2}, we have $\rainbow(d) \in O(d\log d)$.
\end{corollary}
By Corollary \ref{cor:thm:MTH}, we have the upper bound  of $O(d\log d)$ on $\rainbow(d)$. Using this upper bound in Theorem \ref{thm:rainbow2EFX} we obtain a new  upper bound on the number of discarded goods in $\efx$ allocations.
\begin{corollary}
	By choosing $\beta = 1$ and $\gamma = 1$  in Theorem \ref{thm:rainbow2EFX}, For every constant $\varepsilon \in (0,1/2]$, we can find a $(1-\varepsilon)$-$\efx$ allocation with $O_{\epsilon}(\sqrt{n \log n})$ number of discarded goods.
\end{corollary}

\section{Permutation Rainbow Cycle}
In this section, we consider the Permutation Rainbow Cycle problem. For an integer $d>0$, define $\bothpi{\ell}{d}$, $\dpi{d}$, and $\partpi{\ell}$ respectively as subsets of $\bothomega{\ell}{d}$, $\domega{d}$, and $\partomega{\ell}$  consisting all graphs $\graph$ with the additional property that each vertex in $\graph$ has exactly one outgoing edge to every other part. 
Also, we define $\rainbowperm(d)$ as the largest $k$ such that a $k$-partite graph exists in $\dpi{d}$, i.e.,
\[\rainbowperm(d) = \max_{\graph \in \dpi{d}} \#(\graph).\]
Our result in this section is an improved upper bound on $\rainbowperm(d)$ for every $d\geq 3$.
Our method slightly improves the method of Akrami \etal \shortcite{Akrami2022}, wherein the authors prove the upper bound of $2d-2$ on $\rainbowperm(d)$.
Throughout this section, we show that for $d \ge 4$ we have $\rainbowperm(d) \le 2d-4$.
In order to prove this bound, first in Theorem \ref{main2} we show that for $d \ge 3$, we have $\rainbowperm(d) \le 2d-3$. In the proof of Theorem \ref{main2}, we use the idea of constructing a sequence with certain properties. This idea has been previously used by Akrami \etal \shortcite{Akrami2022} to prove the upper bound of $2d-2$. Here, we strengthen the assumptions on the sequence. Lemma \ref{helping_node} plays a key role in route to proving our upper bound. 

We next show how we can incorporate $\ech(\ell)$ in the proof to improve the upper bound to $2d-4$. Later in Section \ref{section:exp}, we discuss the possibility of obtaining better upper bounds on $\rainbowperm(d)$ via a more effective incorporation of  $\ech(\ell)$ in the proof. While it might be possible to obtain $2d-c$ upper bound for $c>4$ with the same idea, we show that it is not possible to obtain an upper bound in the form of $d+c$ for a constant $c>0$ using the same method.


\begin{theorem}\label{main2}
	$\rainbowperm(d) \le 2d-3$ for $d \ge 3$.
\end{theorem}

As a contradiction, suppose there is a graph $\graph$ in $\dpi{d}$ consisting of at least $2d-2$ parts, i.e., $\#(\graph) \ge 2d-2$. We denote by $v_{i,j}$ the $j$'th vertex in the $i$'th part of $\graph$.

The first important step in order to improve the previous result is stated in Lemma \ref{helping_node}.
In this lemma, we show that for every vertex $v$ there is a vertex $u$ with an outgoing edge to $v$, such that no other vertex has outgoing edge to both $v$ and $u$.

\begin{lemma} \label{helping_node}
	For each vertex $\vertex$, there exists some vertex $\vertexx$ with an outgoing edge to $\vertex$ such that
	for any vertex $w$ with an outgoing edge to $\vertex$, $w$ does not have an outgoing edge to $\vertexx$.
\end{lemma}

\begin{proof}
	Consider the induced subgraph of $\graph$, consisting $v$ and all the vertices that have an outgoing edge to $v$. Recall that by definition of $\rainbowperm(d)$, none of these vertices belong to the same part of $\graph$. In this subgraph, the in-degree of at least one vertex is zero; otherwise we have a rainbow cycle in $\graph$. The vertex with in-degree zero in the induced subgraph satisfies the condition of Lemma \ref{helping_node}.
\end{proof}

Consider vertex $\bigvertex{1,1}$.  We know that in every other part, there exists a vertex with an outgoing edge to $\bigvertex{1,1}$. Without loss of generality, we assume that for every $j$, vertex $\bigvertex{j,1}$ is the vertex with an outgoing edge to $\bigvertex{1,1}$.

Also, by Lemma \ref{helping_node}, we know that there exists an index $k$ such that $\bigvertex{k,1}$ has no incoming edge from any $\bigvertex{k',1}$ for $k' \notin \{1,k\}$. Again, without loss of generality, we suppose that $k=2$. Therefore, we have that for every $i>1$, vertex $\bigvertex{i,1}$ has an outgoing edge to  $\bigvertex{1,1}$ and for every $i>2$, $\bigvertex{i,1}$ does not have an outgoing edge to $\bigvertex{2,1}$. By definition, we know that for every $i \neq 2$, there exists a vertex in part $\ppart_{i}$ with an outgoing edge to   $\bigvertex{2,1}$. Without loss of generality, we suppose that for every $i>2$, this vertex in part $\ppart_i$ is $\bigvertex{i,2}$.
                                                                                                                                       
\begin{definition}
	Consider a sequence of indices $\sigma = \sigma_1, \sigma_2, ..., \sigma_k$, such that $\sigma_1 = 1$. Given $\sigma$, we say a vertex $\bigvertex{\sigma_i, j}$ is $\sigma$-\reachable~if there exists a rainbow path from $\bigvertex{1,1}$ to $\bigvertex{\sigma_i,j}$  in $\graph[\{\ppart_{\sigma_1},\ppart_{\sigma_2},\ldots,\ppart_{\sigma_k}\}]$.
	Moreover, we say an edge in $\graph[\{\ppart_{\sigma_1},\ppart_{\sigma_2},\ldots,\ppart_{\sigma_k}\}]$ is $\sigma$-rightward if it is of the form $(\vertex_{\sigma_{j,k}},\vertex_{\sigma_{{j'},k'}})$ where $j<j'$. A vertex $\bigvertex{\sigma_i, j}$ is $\sigma$-\rightreachable~if there exists a rainbow path from $\bigvertex{1,1}$ to $\bigvertex{\sigma_i,j}$ via $\sigma$-rightward edges.
\end{definition}

As we mentioned before, it is sufficient to show if $\graph$ contains at least $2d-2$ parts (and $d \ge 3$), then we have a rainbow cycle in $\graph$. We use induction to prove this claim.  For the base case $d=3$, it has already shown in \cite{chaudhury2021improving} that  $\rainbowperm(3) = 3$, which means $\rainbowperm(3) \le 2 \times 3 - 3$. Now, suppose that the claim holds for every $d'<d$ and our goal is to prove the claim for $d$.  As a contradiction, we suppose that $\graph$ does not admit any rainbow cycle.
We start by proving Lemma \ref{list}. 

\begin{lemma}
	\label{list}
	There exists a sequence of form  $\sigma = \sigma_1,\sigma_2,\ldots,\sigma_{2d-3}$ such that for every $1 \leq i \leq 2d-3$, we have $\sigma_i \in [1,2d-2]$ and the following properties hold:  
	\begin{itemize}
		\item $\sigma_1 = 1$.
		\item For every $1 \leq i \leq 2d-3$, we have $\sigma_i \neq 2$.
		\item For every $2 \leq i \leq 2d-4$, there are $\lceil \frac{i}{2}\rceil$ $\sigma$-\rightreachable~vertices in $\ppart_{\sigma_i}$.  
		\item There are $d-2$ $\sigma$-\rightreachable~vertices in $\ppart_{\sigma_{2d-3}}$. 
	\end{itemize}
\end{lemma}

\begin{figure}
	\centering
	\includegraphics[scale=0.7]{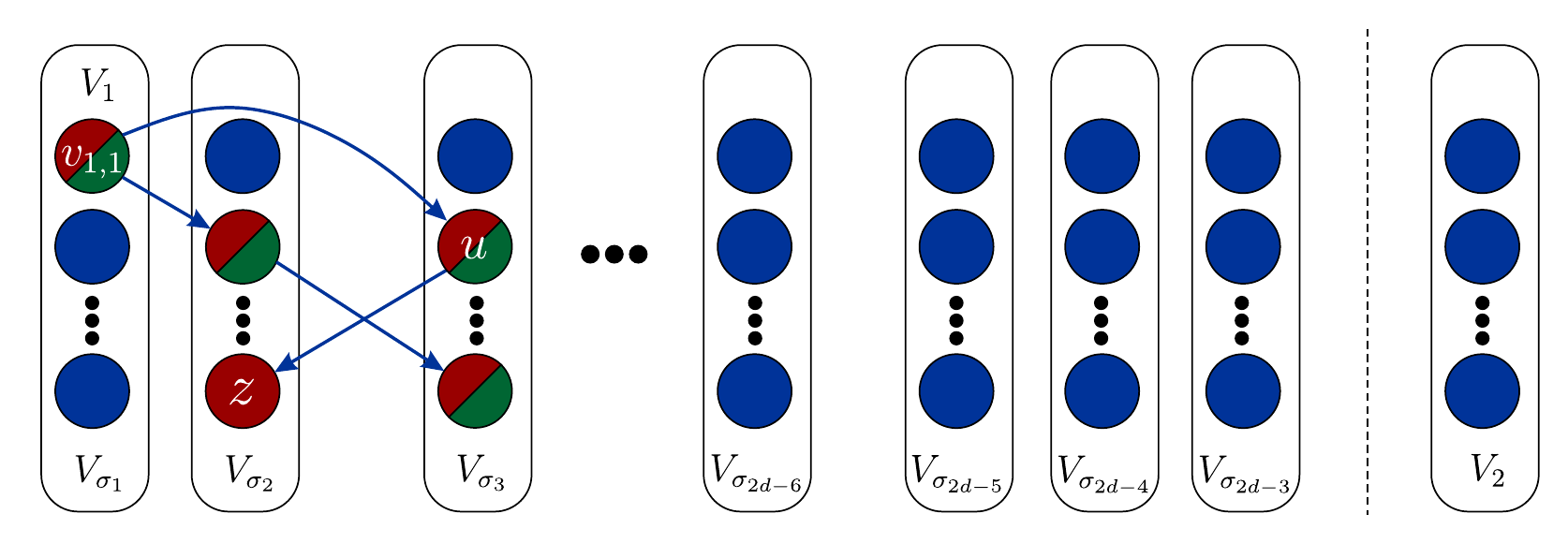}
	\caption{An illustration of the final setting of Lemma \ref{list}. Red vertices are \reachable~and green vertices are \rightreachable. Vertex $z$ is \reachable~but not \rightreachable~since the path from $\bigvertex{1,1}$ uses the edge from $u$ to $z$ which is not rightward.} \label{fig:specialcase1}
\end{figure}

\begin{proof}
		In order to prove that such a sequence exists, first we prove that for each $1 \le i \le d-2$, we can find a sequence $\sigma_1, \sigma_2, ..., \sigma_{2i-1}$, such that $\sigma_1=1$ and for each $1 \leq j \leq 2i-1$, $\lceil \frac{j}{2}\rceil$ different vertices of part $\ppart_{\sigma_j}$ are $\sigma$-\rightreachable. To prove this, we use induction on $i$. 
The base case $i = 1$ is trivial: just consider the sequence which consists of $\sigma_1=1$. Note that $\bigvertex{1,1} \in \ppart_{1}$ and this vertex is $\sigma$-\rightreachable~with a path of length 0 (one vertex). Now, suppose that the claim holds for every $i'<i$ and we want to prove it for $i$. Let 
$\sigma_1, \sigma_2, ..., \sigma_{2i-3}$ be the desired sequence for $i-1$. If we put parts $\ppart_{\sigma_j}$ for $1 \leq j \leq 2d-3$ and $\ppart_{2}$ aside, we have $2d-2-(2i-3+1)=2(d-i)$ remaining parts. For each of the remaining parts such as $\pparttt$, we already know that $i-1$ vertices of $\pparttt$ can be $\sigma$-\rightreachable, since $i-1$ vertices of $\ppart_{\sigma_{2i-3}}$ are $\sigma$-\rightreachable, and from those $i-1$ vertices, we can go to $i-1$ vertices of $\pparttt$ via a direct edge. We mark all those $i-1$ vertices in the remaining parts. Each of these $2(d-i)$ remaining parts has $d-i+1$ unmarked vertices. 

\begin{figure}
	\centering
	\includegraphics[scale=0.65]{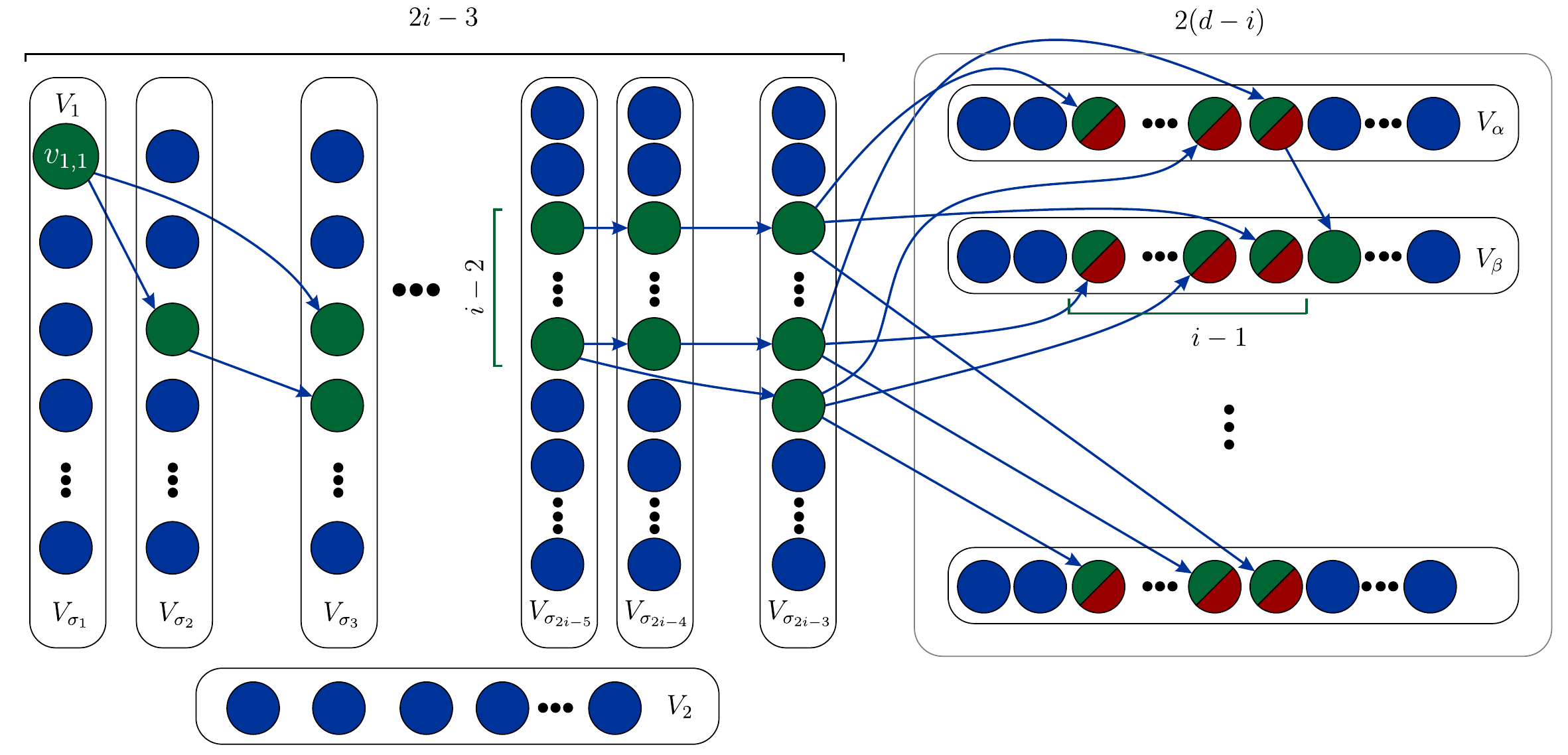}
	\caption{This figure is an illustration of adding two new indices at the end of the sequence. The green vertices in $2i-3$ parts in the left are $\sigma$-\rightreachable. In the $2(d-i)$ remaining parts, vertices with red color are marked vertices and vertices with green color are potentially $\sigma$-\rightreachable, i.e., if we add their index at the end of the sequence, then these vertices are $\sigma$-\rightreachable. Moreover, parts $\ppart_\alpha$ and $\ppart_\beta$ are illustrated in this figure. There is an outgoing edge from a marked vertex in $\ppart_\alpha$ to an unmarked vertex in $\ppart_\beta$, which makes an additional vertex in $\ppart_\beta$ potentially $\sigma$-\rightreachable. Thus, we add $\alpha$ and $\beta$ to the end of the sequence as $\sigma_{2i-2}$ and $\sigma_{2i-1}$.} \label{fig:prmutationsequencelemma1}
\end{figure}

Now, if for two of these parts, namely $\ppart_\alpha$ and $\ppart_\beta$, there exists a directed edge from a marked vertex of $\ppart_\alpha$ to an unmarked vertex of $\ppart_\beta$,  then we can add $\alpha$ and $\beta$ to the end of the sequence, i.e., $\sigma_{2i-2} = \alpha$ and $\sigma_{2i-1} = \beta$. Note that at least $i$ vertices of $\ppart_\beta$ are $\sigma$-\rightreachable: $i-1$ marked vertices and one unmarked vertex which has an incoming edge from a marked vertex of $\ppart_\alpha$. You can see an illustration of this part of the proof in Figure \ref{fig:prmutationsequencelemma1}. 		

Therefore, we can assume that there is no outgoing edge from a marked vertex to an unmarked vertex among these $2(d-i)$ parts. Note that, since the edges of each part form a permutation, the reverse is also true: there is no outgoing edge from an unmarked vertex to a marked vertex.  Now, let $\graphh$ be the induced subgraph of $\graph$ consisting the unmarked vertices of these $2(d-i)$ parts. Since each of these parts has $i-1$ marked vertices in $\graph$, each part in $\graphh$ consists of $d-i+1$ vertices. If $\graphh$ admits no rainbow cycle, we must have $\graphh \in \dpi{d-i+1}$. However, since $i\leq d-2$, we have $d-i+1\geq 3$. Based on induction hypothesis, we know $\rainbowperm(d-i+1)\leq 2(d-i)-1$, which contradicts that $\graphh$ contains $2(d-i)$ parts.

Therefore, we can form a sequence $\sigma_1, \sigma_2, ..., \sigma_{2d-5}$. Since $\graph$ has $2d-2$ parts, other than the parts pointed to in the sequence and part $\ppart_2$ we have two remaining parts, namely $\ppart_\alpha$ and $\ppart_\beta$.
We can add $\alpha$ and $\beta$ to the end of the sequence as $\sigma_{2d-4}$ and $\sigma_{2d-3}$ because $d-2$ vertices of $\ppart_{\sigma_{2d-5}}$ are $\sigma$-\rightreachable~and each of them has an outgoing edge to a unique vertex in $\ppart_{\sigma_{2d-4}}$ and $\ppart_{\sigma_{2d-3}}$. Therefore, there are $d-2$ $\sigma$-\rightreachable~vertices in $\ppart_{\sigma_{2d-4}}$ and $\ppart_{\sigma_{2d-3}}$. As a result, we can form a sequence of form $\sigma_1, \sigma_2, ..., \sigma_{2d-3}$ consisting of $2d-3$ parts, which is the desired sequence.

\end{proof}

Let $\sigma$ be the sequence that satisfies the properties of Lemma \ref{list}. In Lemma \ref{label_one_and_two}, we prove another property for such a sequence.

\begin{lemma} \label{label_one_and_two}
	For every sequence $\sigma$ with properties mentioned in Lemma \ref{list} and every $2 \leq i \leq 2d-3$, vertices $\bigvertex{\sigma_i,1}$ and $\bigvertex{\sigma_i,2}$ are not $\sigma$-\reachable.
\end{lemma}

\begin{proof}
	As a contradiction, if $\bigvertex{\sigma_i,1}$ is $\sigma$-\reachable, since $\bigvertex{\sigma_i,1}$ has an outgoing edge to $\bigvertex{1,1}$, then we have a rainbow cycle in $\graph$.
	On the other hand, if $\bigvertex{\sigma_i,2}$ is $\sigma$-\reachable, since there is an outgoing edge from $\bigvertex{\sigma_i,2}$ to $\bigvertex{2,1}$, and there is an outgoing edge from $\bigvertex{2,1}$ to $\bigvertex{1,1}$, and part $\ppart_2$ is outside of the sequence, then we have a rainbow cycle in $\graph$.
\end{proof}

\begin{definition} \label{onecorr}
	If we consider $S$ as a subset of $\{1, 2, ..., d\}$, part $\ppart_i$ is \oneway~$S$-\dep~to part $\ppart_j$, if and only if for each vertex $\bigvertex{i,k}$ such that $k \in S$, it has an outgoing edge to $\bigvertex{j,l}\in \ppart_j$, such that $l \in S$.
\end{definition}

\begin{definition} \label{corr}
	Parts $\ppart_i$ and $\ppart_j$ are $S$-\dep, if and only if:
	\begin{itemize}
		\item $\ppart_i$ is \oneway~$S$-\dep~to $\ppart_j$
		\item $\ppart_j$ is \oneway~$S$-\dep~to $\ppart_i$
	\end{itemize}
\end{definition}
\begin{figure}
	\centering
	\includegraphics[scale=0.7]{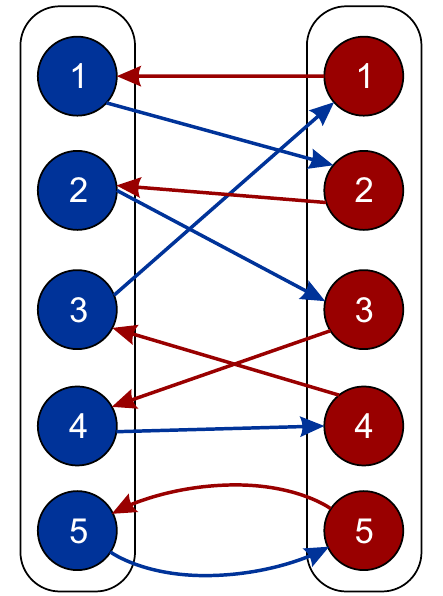}
	\caption{In this example, the blue part is \oneway~$\{1,2,3\}$-\dep~to the red part. Moreover, the blue part and the red part are $\{1,2,3,4\}$-\dep, since the blue part is \oneway~$\{1,2,3,4\}$-\dep~to the red part and vice versa.} \label{fig:permutation:scorresponding}
\end{figure}
See Figure \ref{fig:permutation:scorresponding} for an illustrative example.

\begin{lemma} \label{three_pairwise}
	There are three pairwise $\{1,2\}$-\dep~parts.
\end{lemma}

\begin{proof}
	Consider the last $5$ elements of the sequence, i.e., $\sigma_{2d-7}$, $\sigma_{2d-6}$, $\sigma_{2d-5}$, $\sigma_{2d-4}$, $\sigma_{2d-3}$. For brevity, we define $\rho = \sigma_{2d-7}$, $\alpha = \sigma_{2d-6}$, $\beta = \sigma_{2d-5}$, $\gamma = \sigma_{2d-4}$, $\delta = \sigma_{2d-3}$.
	
	For each of the parts $\ppart_{\beta}$, $\ppart_{\gamma}$, and $\ppart_{\delta}$, there are two vertices which are not $\sigma$-\rreachable. By Lemma \ref{label_one_and_two}, these vertices have labels $1$ and $2$, i.e., $\bigvertex{\beta, 1}, \bigvertex{\beta,2}, \bigvertex{\gamma,1}$, etc. 
	Moreover, $\ppart_{\rho}$ has $d-3$ vertices which are $\sigma$-\rightreachable. From these $d-3$ vertices, there are edges to $d-3$ different vertices of $\ppart_{\alpha}$, $\ppart_{\beta}$, $\ppart_{\gamma}$, and $\ppart_{\delta}$. Indeed, these $d-3$ vertices for each part are $\sigma$-\rightreachable. As we mentioned, the vertices with labels $1$ and $2$ in these parts are not $\sigma$-\reachable. Since there are $d$ vertices in each part, there is one other vertex in each of these parts. We assume without loss of generality that the label of this vertex in each part is 3, i.e., $\bigvertex{\alpha,3}, \bigvertex{\beta,3}$, etc.

	Next, we go through a series of claims, which will lead to our conclusion. For brevity, here we provide a list of these claims along with their proofs.
	\begin{enumerate}[label= (\roman*)]
		\item \label{c1}$\ppart_{\gamma}$ and $\ppart_{\delta}$ are $\{1,2\}$-\dep.
		
		\begin{subproof2}
			As a contradiction, \mywlog, we can assume that there is a vertex $\bigvertex{\delta,x}$ such that $x \notin \{1,2\}$ and it has an outgoing edge either to $\bigvertex{\gamma,1}$ or $\bigvertex{\gamma,2}$. Since $\bigvertex{\delta,x}$ is $\sigma$-\reachable~and the rainbow path from $\bigvertex{1,1}$ to it does not pass over any vertex in $\ppart_{\gamma}$, $\bigvertex{\gamma,1}$ or $\bigvertex{\gamma,2}$ are $\sigma$-\reachable. Therefore, by Lemma \ref{label_one_and_two}, we have a rainbow cycle.
		\end{subproof2}
		\item \label{c2} $\ppart_\beta$ is \oneway~$\{1,2\}$-\dep~with $\ppart_\gamma$ and $\ppart_\delta$.
		
		\begin{subproof2}
			As a contradiction, \mywlog, there is a vertex $\bigvertex{\beta,x}$ such that $x \notin \{1,2\}$ and it has an outgoing edge either to $\bigvertex{\gamma,1}$ or $\bigvertex{\gamma,2}$. Since $\bigvertex{\beta,x}$ is $\sigma$-\rightreachable, $\bigvertex{\gamma,1}$ or $\bigvertex{\gamma,2}$ are $\sigma$-\reachable. Therefore, by Lemma \ref{label_one_and_two}, we have a rainbow cycle.
		\end{subproof2}
		\item \label{c3} $\ppart_\gamma$ and $\ppart_\delta$ are not $\{1,2,3\}$-\dep.
		
		\begin{subproof2}
			As a contradiction, suppose $\ppart_\gamma$ and $\ppart_\delta$ are $\{1,2,3\}$-\dep. Since they are $\{1,2\}$-\dep~as well, $\bigvertex{\gamma,3}$ and $\bigvertex{\delta,3}$ should have direct edge to each other. Therefore, we have a rainbow cycle consist of two vertices.
		\end{subproof2}
		
		Note that, since $\ppart_\gamma$ and $\ppart_\delta$ are $\{1,2\}$-\dep, but  not $\{1,2,3\}$-\dep, either there is a vertex $\bigvertex{\delta,x}$ where $x \notin \{1,2,3\}$ and $\bigvertex{\delta,x}$ has an  outgoing edge to $\bigvertex{\gamma,3}$, or there is a vertex $\bigvertex{\gamma,x}$ such that $x \notin \{1,2,3\}$ and $\bigvertex{\gamma,x}$ has an outgoing edge to $\bigvertex{\delta,3}$. \Wlog, here we assume the first assumption holds and therefore, vertex $\bigvertex{\delta,x}$ ($x \notin \{1,2,3\}$) has an outgoing edge to $\bigvertex{\gamma,3}$.
		
		\begin{assumption} \label{assump1}
			There exists a vertex $\bigvertex{\delta,x}$ with $x \notin \{1,2,3\}$, such that  $\bigvertex{\delta,x}$ has an outgoing edge to $\bigvertex{\gamma,3}$.
		\end{assumption}
		
		\begin{figure*}
			\centering
			\begin{subfigure}{.5\textwidth}
				\centering
				\includegraphics[scale=0.7]{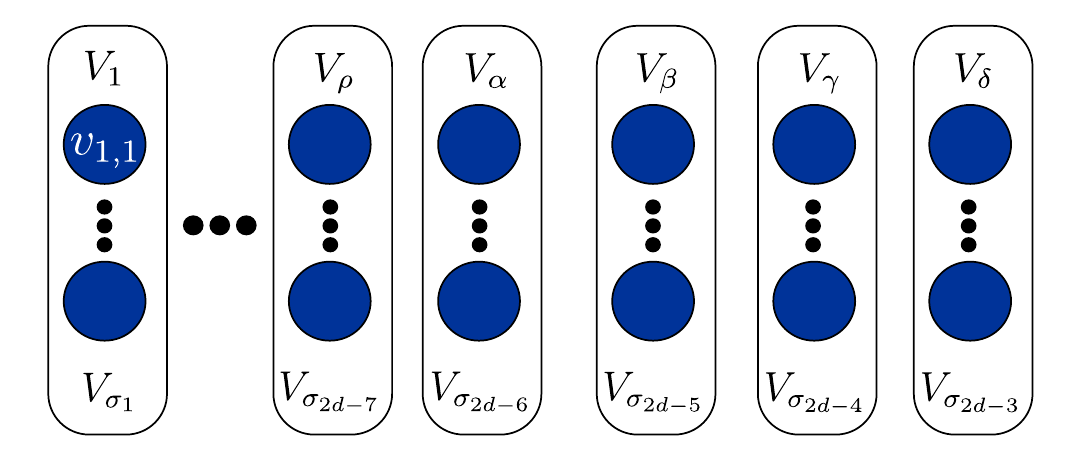}
				\caption{} \label{fig:specialcase2}
			\end{subfigure}%
			\begin{subfigure}{.5\textwidth}
				\centering
				\includegraphics[scale=0.7]{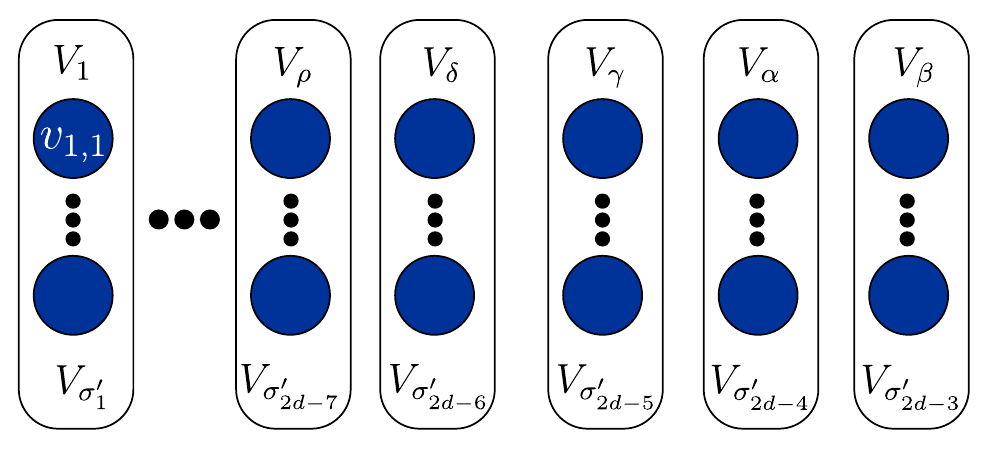}
				\caption{} \label{fig:specialcase3}
			\end{subfigure}
			\caption{Two different orderings of $\ppart_\alpha$, $\ppart_\beta$, $\ppart_\gamma$, $\ppart_\delta$, which we use in the proof of Lemma \ref{three_pairwise}.}
		\end{figure*}
		
		Considering Assumption \ref{assump1}, if we make a new sequence $\sigma'$ where $\sigma'_{2d-6} = \delta$, $\sigma'_{2d-5} = \gamma$, $\sigma'_{2d-4} = \alpha$, $\sigma'_{2d-3} = \beta$, and $\sigma'_i = \sigma_i$ for $1 \leq i \leq 2d-7$ (Figures \ref{fig:specialcase2} and \ref{fig:specialcase3} illustrate the previous formation and the new formation), then the new sequence holds the criteria mentioned in Lemma \ref{list} as well. This means that $d-2$ vertices of $\ppart_\gamma$ (and consequently $d-2$ vertices of $\ppart_\alpha$ and $\ppart_\beta$) are $\sigma'$-\rightreachable, because if we go from $\ppart_\rho$ to $\ppart_\gamma$, all vertices other than $\bigvertex{\gamma,1}$, $\bigvertex{\gamma,2}$, and $\bigvertex{\gamma,3}$ are $\sigma'$-\rightreachable. Additionally, for some $x \notin \{1,2,3\}$, there is an edge from $\bigvertex{\delta,x}$ to $\bigvertex{\gamma,3}$. This edge makes $\bigvertex{\gamma,3}$ $\sigma'$-\rightreachable~because the first $2d-7$ elements of $\sigma'$ and $\sigma$ are the same and $\bigvertex{\delta,x}$ is $\sigma'$-\rightreachable~as well as $\sigma$-\rightreachable~when $x \notin \{1,2,3\}$. Therefore, $d-2$ vertices of parts $\ppart_\gamma$, $\ppart_\alpha$, and $\ppart_\beta$ are $\sigma'$-\rightreachable, and since the first $2d-7$ elements are equal in $\sigma$ and $\sigma'$,  the new sequence $\sigma'$ has the properties we want as well.

		\item \label{c4} $\ppart_\alpha$ and $\ppart_\beta$ are $\{1,2\}$-\dep, and by Assumption \ref{assump1}, $\ppart_\gamma$ is \oneway~$\{1,2\}$-\dep~to $\ppart_\alpha$ and $\ppart_\beta$. 
		
		\begin{subproof2}
			$\ppart_\alpha$ and $\ppart_\beta$ are $\{1,2\}$-\dep~because now they are the last two vertices of the sequence, and we can have the same conclusion that we had for $\ppart_\gamma$ and $\ppart_\delta$ in Claim \ref{c1}. Also, $\ppart_\gamma$ is \oneway~$\{1,2\}$-\dep~to $\ppart_\alpha$ and $\ppart_\beta$ with the same conclusion we had in Claim \ref{c2}. 
		\end{subproof2}
		
		\item \label{c5} Considering Assumption \ref{assump1}, $\ppart_\beta$ and $\ppart_\gamma$ are $\{1,2\}$-\dep, and they are not $\{1,2,3\}$-\dep. Also, $\ppart_\alpha$ and $\ppart_\delta$ are $\{1,2\}$-\dep, and they are not $\{1,2,3\}$-\dep.
		
		\begin{subproof2}
			Since $\ppart_\beta$ is \oneway~$\{1,2\}$-\dep~to $\ppart_\gamma$ and vice versa, they are $\{1,2\}$-\dep. Therefore, with the same conclusion we used for $\ppart_\gamma$ and $\ppart_\delta$ in Claim \ref{c3}, they cannot be $\{1,2,3\}$-\dep.
			When $\ppart_\beta$ and $\ppart_\gamma$ are not $\{1,2,3\}$-\dep, we can form a new sequence $\sigma''$ such that $\beta$ and $\gamma$ come after $\rho$, and $\alpha$ and $\delta$ come as the last two elements. We can argue that this new sequence has the properties mentioned in Lemma \ref{list} as well, because $\beta$ and $\gamma$ are not $\{1,2,3\}$-\dep~and we can have the same reasoning we had for $\sigma'$ again. As a result, $\ppart_\alpha$ and $\ppart_\delta$ are the last two parts and, yet again, we can say they are $\{1,2\}$-\dep~and they cannot be $\{1,2,3\}$-\dep.
		\end{subproof2}
		
		\item \label{c6} Considering Assumption \ref{assump1}, either $\ppart_\alpha$ and $\ppart_\gamma$, or $\ppart_\beta$ and $\ppart_\delta$ are $\{1,2\}$-\dep.
		
		\begin{subproof2}
			Since $\ppart_\alpha$ and $\ppart_\delta$ are not $\{1,2,3\}$-\dep, we can form a new sequence $\sigma'''$ and put $\alpha$ and $\delta$ after $\rho$, and put $\beta$ and $\gamma$ as the last two elements. This new sequence has the properties mentioned in Lemma \ref{list} as well. Now, we have two cases (depends on whether there exists a vertex $\bigvertex{\alpha, x}$ where $x \notin \{1,2,3\}$ and $\bigvertex{\alpha, x}$ has an outgoing edge to $\bigvertex{\delta,3}$, or there exists a vertex $\bigvertex{\delta, x}$ where $x \notin \{1,2,3\}$ and $\bigvertex{\delta, x}$ has an outgoing edge to $\bigvertex{\alpha,3}$):
			\begin{itemize}
				\item  If the last 5 elements of $\sigma'''$ are $\rho, \delta, \alpha, \beta, \gamma$, then $\ppart_\alpha$ is \oneway~$\{1,2\}$-\dep~to $\ppart_\gamma$, and $\ppart_\gamma$ was \oneway~$\{1,2\}$-\dep~to $\ppart_\alpha$ as well; therefore, they are $\{1,2\}$-\dep.
				\item If the last 5 elements of $\sigma'''$ are $\rho, \alpha, \delta, \beta, \gamma$, then $\ppart_\delta$ is \oneway~$\{1,2\}$-\dep~to $\ppart_\beta$, and  $\ppart_\beta$ was \oneway~$\{1,2\}$-\dep~to $\ppart_\delta$ as well; therefore, they are $\{1,2\}$-\dep.
			\end{itemize}
		\end{subproof2}
	\end{enumerate}
	
	Now, note that by Claim \ref{c6}, we have two cases:
	\begin{itemize}
		\item If $\ppart_\alpha$ and $\ppart_\gamma$ are $\{1,2\}$-\dep, since we have already proved that $\ppart_\beta$ is $\{1,2\}$-\dep~to both $\ppart_\alpha$ and $\ppart_\gamma$, we have that parts $\ppart_\alpha$, $\ppart_\beta$, and $\ppart_\gamma$ are pairwise $\{1,2\}$-\dep.
		\item If $\ppart_\beta$ and $\ppart_\delta$ are $\{1,2\}$-\dep, since we have already proved that $\ppart_\gamma$ is $\{1,2\}$-\dep~with both of parts $\ppart_\beta$ and $\ppart_\delta$, we have that $\ppart_\beta$, $\ppart_\gamma$, and $\ppart_\delta$ are pairwise $\{1,2\}$-\dep.
	\end{itemize}
	This completes the proof of Lemma \ref{three_pairwise}.
\end{proof}	

Finally, note that by Lemma \ref{three_pairwise}, there exists three pairwise $\{1,2\}$-\dep~parts in $G$. Therefore, if we consider the induced subgraph $G'$ of $G$ containing vertices with indices $1,2$ in these three parts, since $G \in \dpi{d}$, $G'$ must belong to $\dpi{2}$. However, we know that $\rainbowperm(2)=2$, which means that $\graphh$ cannot have more than two parts.
This contradiction shows that, if graph contains at least $2d-2$ parts, then it has a rainbow cycle. Hence, $\rainbowperm(d) \le 2d-3$.

\paragraph{Improving the upper-bound to $2d-4$.}

We end this section by a discussion on how we can improve the upper bound to $2d-4$. Recall the definition of $\ech(\ell)$. As we show in Section \ref{section:exp}, we have $\ech(4) = 7$. 
This means that for every set $W$ of parts with $|W|=4$, for any vertex $\vertex \in G[W]$, there are at least 7 other vertices that have a rainbow path to $\vertex$ in $G[W]$.  Since the graph is a permutation graph, the inverse direction is also true: for any vertex $\vertex \in G[W]$, $\vertex$ has rainbow paths to at least $7$ different vertices in $G[W]$. We use this fact to decrease the upper bound on  $\rainbowperm(d)$ by one.

Consider part $\ppart_1$ and three arbitrary parts other than $\ppart_2$ (i.e., $\ppart_3$, $\ppart_4$, and $\ppart_5$). It is guaranteed that vertex $\bigvertex{1,1}$ has rainbow paths to at least $7$ different vertices in these three parts. 
Therefore, by the pigeonhole principle, vertex $\bigvertex{1,1}$ has rainbow paths to $3$ vertices in one of these parts. Assume without loss of generality that this part is $\ppart_5$. Now, we create a shortcut in the sequence by replacing $\sigma_2$, $\sigma_3$, $\sigma_4$, $\sigma_5$ with $3, 4, 5$. Note that though parts $\ppart_{3}$ and $\ppart_4$ might violate the properties of the sequence (e.g., rainbow paths to $\ppart_5$ are not necessarily $\sigma$-\rightreachable), but $\ppart_5$ can be treated the same way as $\ppart_{\sigma_5}$ in the previous sequence, which was the first part with $3$ $\sigma$-\rightreachable~vertices. Therefore, we can continue constructing the sequence from $\ppart_5$ in the same way as we construct the sequence (first, add $\sigma_6$ and $\sigma_7$, next $\sigma_8$ and $\sigma_9$, and so on). This way, we save one part in the sequence and therefore, the length of the sequence is reduced to $2d-4$.  Hence, we can conclude that if we have $\max(4, 2d-4)$ parts, then we have a rainbow cycle. As a result, $\rainbowperm(d) \le 2d-4$ for $d \ge 4$.
\section{Experiments}\label{section:exp}
In order to evaluate $\ech(\ell)$, we performed a set of experiments to calculate $\ech(\ell)$ for small values of $\ell$.
Our algorithm inputs $\ell,x$ and performs an exhaustive search to find a counter-example for $\ech(\ell)>x$. By the definition of $\ech(\ell)$, this counter-example must have at most $x$ vertices with a rainbow path to a specific vertex $\vertex$. If such an example is found, we have $\ech(\ell)\leq x$. Otherwise, when there is no such example, we can imply that $\ech(\ell)>x$. The overall result of running this experiment is shown in Table \ref{tab:table1}.

\begin{table}[!ht]
	\begin{center}
		\begin{tabular}{c|c|c} 
			\textbf{$\ell$} & \textbf{Lower bound} & \textbf{Upper bound}\\
			\midrule
			2 & 1 & 1\\
			3 & 3 & 3\\
			4 & 7 & 7\\
			5 & 11 & 11\\
			6 & 15 & 17\\
			7 & - & 25\\
		\end{tabular}
	\end{center}
	\caption{Lower bounds and upper bounds on $\ech(\ell)$ obtained by the	 experiments.}\label{tab:table1}
\end{table}

As you can see in Table \ref{tab:table1}, for $ 2 \leq \ell \leq 5$, the exact value of $\ech(\ell)$ is determined by the experiments. Also, for $\ell = 6,7$, our experiments provide an upper bound on $\ech(\ell)$.
Recall that by Theorem \ref{MTH}, we have $\ech(\ell)\in \Omega(\ell^{2}/\ln \ell)$. In Lemma \ref{OBS_EX}, we prove an upper bound of $O(n^2)$ on $\ech(\ell)$.
\begin{figure}
	\centering
	\includegraphics[scale=0.8]{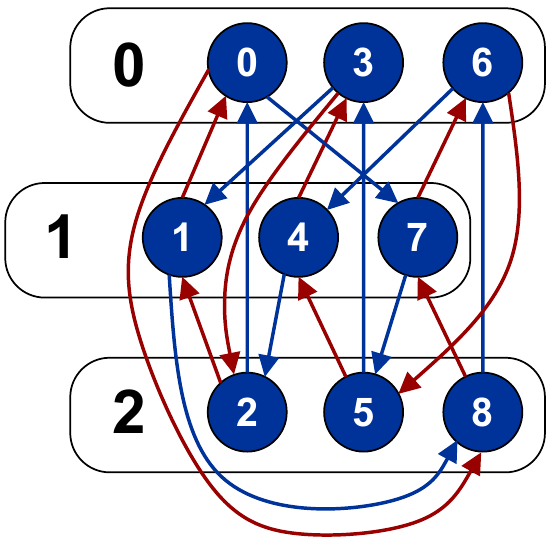}
	\caption{An illustration of the example which is defined in Lemma \ref{OBS_EX} for $\ell=3$. The edge colors in the figure differentiate the edges which connect two vertices with differences $1$ and $2$ modulo $\ell^2=9$.}
	\label{figobs71}
\end{figure}
\begin{figure}
	\centering
	\includegraphics[scale=0.8]{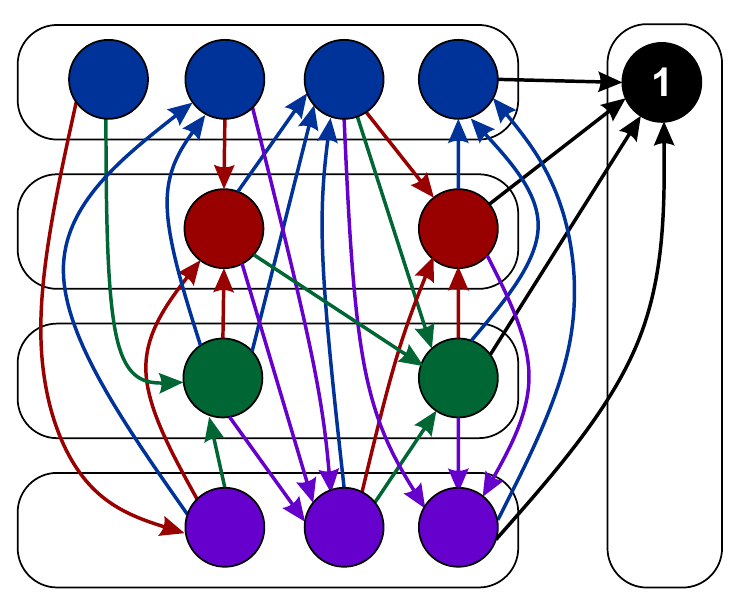}
	\caption{In this figure, you can find a compact form of a graph that shows $\ech(5)=11$. Due to lack of space and for convenience,  		
		here we only show  the induced subgraph of the vertices that have a rainbow path to vertex $1$. Let $G$ be the graph in this figure. In order to construct the entire Graph, one can proceed as follows. Merge $G$ and the graph constructed in the proof of Lemma \ref{OBS_EX} for $\ell=5$ ($G'$). The vertices of each part in the union graph are the union of the vertices in the corresponding parts in $G$ and $G'$. Similarly, the edges in the union graph are the union of the edges in $G$ and $G'$. In addition, some of the vertices in $G$ do not have  incoming edges from some other parts. For such pairs of vertices and parts, we choose an arbitrary vertex from the corresponding part of $G'$ and add a directed edge to that vertex. 
		%
	} 
	\label{fig:h5-sample-11}
\end{figure}

\begin{lemma}\label{OBS_EX}
	We have  $\ech(\ell)\leq (\ell-1)(\ell-2)+1$.
\end{lemma}
\begin{proof}
	We construct a graph $\graph$ with parts $\{\ppart_0,\ldots,\ppart_{\ell-1}\}$ and  vertices $\vertex_0, \ldots, \vertex_{\ell^2-1}$. For each $0\le i < \ell^2$, we have $\vertex_i \in \ppart_{i\%\ell}$, and the edges of $\graph$ are as follows: each $v_i$ has incoming edges from $v_{(i+1)\% \ell^2}, \ldots, v_{(i+\ell-1)\% \ell^2}$.  In Figure \ref{figobs71}, you can see an illustration of $\graph$ for $\ell=3$. All the vertices have an incoming edge from every other part. Furthermore, since rainbow paths have at most $\ell-1$ edges, and the difference between the indices of two connected vertices is at most $\ell-1$ modulo $\ell$, only $\vertex_0, \ldots, \vertex_{(\ell-1)^2}$ can have a rainbow path to $\vertex_0$. since $\vertex_1, \ldots, \vertex_{(\ell-1)^2}$ do not have an incoming edge from $\vertex_0$, $\vertex_0$ cannot be in any rainbow cycle. Since the vertices in $\graph$ are symmetric, $\graph$ does not have any rainbow cycle. As a result, $\graph \in \partomega{\ell}$.
	
	Moreover, vertices in $\ppart_{0}\setminus \{\vertex_0\}$ cannot have a rainbow path to $\vertex_0$ because they are in the same part. If we exclude $\ppart_{0}$, $(\ell-1)^2+1-(\ell-1)=(\ell-1)(\ell-2)+1$ vertices remain. Therefore, there are at most $(\ell-1)(\ell-2)+1$ vertices that can reach $v_0$ using rainbow paths.
\end{proof}

Note that the upper bound provided by Lemma \ref{OBS_EX} exactly matches the upper bounds for $\ech(2), \ech(3),$ and $\ech(4)$. However, for $\ech(5)$ this upper bound is not tight. In Figure \ref{fig:h5-sample-11}, a tight example for $\ech(5)$ is shown.  
Based on the results extracted from the experiments, our conjecture is as follows. 
\begin{conjecture}\label{K2CJ}
	We conjecture that $\ech(\ell)=\lfloor \frac{\ell^2}{2} \rfloor -1$.
\end{conjecture}

Note that if Conjecture \ref{K2CJ} holds, then using Lemma \ref{MTH2}, we have $\rainbow(d)\in O(d)$. 

We also performed similar experiments to evaluate $\echpi(\ell)$ which is an analogous of $\ech(\ell)$ for permutation graphs. Formally, 
\[\echpi(\ell) = \min_{\graph \in \partpi{\ell}} \min_{\vertex \in \graph} \quad f_\graph(\vertex),\]
where $f_\graph(\vertex)$ is the number of the vertices in $\graph$ that have a rainbow path to $\vertex$.
Interestingly, the results were exactly the same as the previous case stated.
\begin{conjecture}\label{conj:permutation}
	We conjecture that $\echpi(\ell)=\ech(\ell)$.
\end{conjecture}

Similar to Lemma \ref{MTH2}, we can prove a simple relation between  $\echpi(\ell)$ and $\rainbowperm(d)$.
\begin{lemma}\label{lemma:htorainbow:permutation}
	Given that for some $\ell$, $\echpi(\ell)>(d-1)(\ell-1)$, we have $\rainbowperm(d)<\ell$.
\end{lemma}
\begin{proof}
	As a contradiction suppose that $\bothpi{\ell}{d} \neq \emptyset$ and let $\graph \in \bothpi{\ell}{d}$ be a graph with parts $\ppart_{1},\ldots, \ppart_{\ell}$, and let $\vertex \in \ppart_{\ell}$ be a vertex of $G$.
	Since $\echpi(\ell)>(d-1)(\ell-1)$, at least $(d-1)(\ell-1)$ vertices in parts $\{\ppart_{1},\ldots, \ppart_{\ell-1}\}$ have a rainbow path to $\vertex$. Therefore, by the pigeonhole principle, there is a part $\ppartt \in \{\ppart_{1},\ldots, \ppart_{\ell-1}\}$ such that all of its vertices have a rainbow path to $v$. Also, since $v$ has an outgoing edge to $\ppartt$, we have a rainbow cycle in $G$, which is a contradiction.
\end{proof}
\begin{lemma}\label{lemma:ourconjecturetorainbow:permutation}
	 For $d\geq 3$, $\echpi(\ell)\geq \frac{\ell^2}{2} -1$ implies $\rainbowperm(d)\leq 2d-3$.
\end{lemma}
\begin{proof}
	We have
	\[\begin{aligned}
		\echpi(2d-2) &\geq \frac{(2d-2)^2}{2} -1 \\
		& = 2d^2-4d+1\\
		& = (d-1)((2d-2)-1) + (d-2)\\
		& > (d-1)((2d-2)-1). & d>2
	\end{aligned}\] 
	Therefore, by Lemma \ref{lemma:htorainbow:permutation}, $\rainbowperm(d)<2d-2$. Since $\rainbowperm(d)$ is an integer, $\rainbowperm(d)\leq 2d-3$.
\end{proof}

Lemma \ref{lemma:ourconjecturetorainbow:permutation} shows that even if we prove Conjecture \ref{conj:permutation} is correct, we cannot get a better upper bound for $\rainbowperm(d)$ with a simple connection between $\rainbowperm(d)$ and $\echpi(\ell)$. However, we believe proving Conjecture \ref{conj:permutation} would be a good warm-up in the way of  proving Conjecture \ref{K2CJ}. 
\bibliographystyle{abbrv}
\bibliography{draft}

\appendix

\end{document}